%% file: main.tex
\documentclass[a4paper,UKenglish,cleveref, autoref, runningheads]{llncs}
\usepackage[T1]{fontenc}

\usepackage{cite}
\usepackage{lineno}
\usepackage{amsmath,amssymb,amsthm,graphicx,fixmath,tikz,latexsym,color,xspace,hyperref}
\usepackage{tikz}
\usetikzlibrary{calc}
\usepackage{bbold}
\usepackage{mathtools}

\def\Re{\ensuremath{\mathbb{R}}}
\newcommand*{\Positives}{\mathbb{N}}

\newtheorem{observation}{Observation}

\newcommand{\B}{\mathcal{B}}
\newcommand{\C}{\mathcal{C}}
\newcommand{\D}{\mathcal{D}}
\newcommand{\halfplane}{\mathcal{H}}

\begin{document}

\title{Computing a maximum clique in geometric superclasses of disk graphs} 

\titlerunning{Maximum clique in geometric graphs}


\author{Nicolas Grelier \thanks{The author wants to thank Michael Hoffmann for his advice and his help concerning the writing of the paper. He also thanks anonymous reviewers for their helpful comments. Research supported by the Swiss National Science Foundation within the collaborative DACH project Arrangements and Drawings as SNSF Project 200021E-171681.}}

\authorrunning{N. Grelier}

\institute{Department of Computer Science, ETH Z{\"u}rich, Switzerland  
\email{nicolas.grelier@inf.ethz.ch}}

\maketitle

\begin{abstract}

 In the 90's Clark, Colbourn and Johnson wrote a seminal paper where they proved that maximum clique can be solved in polynomial time in unit disk graphs. Since then, the complexity of maximum clique in intersection graphs of $d$-dimensional (unit) balls has been investigated. For ball graphs, the problem is NP-hard, as shown by Bonamy {\em et al.} (FOCS~'18). They also gave an efficient polynomial time approximation scheme (EPTAS) for disk graphs. However, the complexity of maximum clique in this setting remains unknown. In this paper, we show the existence of a polynomial time algorithm for a geometric superclass of unit disk graphs. Moreover, we give partial results toward obtaining an EPTAS for intersection graphs of convex pseudo-disks.

 
\end{abstract}


\section{Introduction}

In an \emph{intersection graph}, every vertex can be represented as a set, such that two vertices are adjacent if and only if the corresponding sets intersect. In most settings, those sets are geometric objects, lying in a Euclidean space of dimension $d$. Due to their interesting structural properties, the intersection graphs of $d$-dimensional balls, called $d$-ball graphs, have been extensively studied. For dimensions $1$, $2$ and $3$, the $d$-ball graphs are called \emph{interval graphs}, \emph{disk graphs} and \emph{ball graphs}, respectively. If all $d$-balls have the same radius, their intersection graphs are referred to as \emph{unit $d$-ball graphs}. The study of these classes has many applications ranging from resource allocation to telecommunications~\cite{bar2001unified,van2009optimization,fishkin2003disk}. 

Many problems that are NP-hard for general graphs remain NP-hard for $d$-ball graphs, with fixed $d\geq 2$. Even for unit disk graphs, most problems are still NP-hard. A famous exception to this rule is the problem of computing a maximum clique, which can be done in polynomial time in unit disk graphs as proved by Clark, Colbourn and Johnson~\cite{clark1990unit}. Their algorithm requires the position of the unit disks to be given, but a robust version of their algorithm, which does not require this condition, was found by Raghavan and Spinrad~\cite{raghavan2001robust}. This is a nontrivial matter as Kang and M{\"u}ller have shown that the recognition of unit $d$-ball graphs is NP-hard, and even $\exists \Re$-hard, for any fixed $d\geq2$~\cite{kang2012sphere}.

Finding the complexity of computing a maximum clique in general disk graphs (with arbitrary radii) is a longstanding open problem. However in 2017, Bonnet {\em et al.}, found a subexponential algorithm and a  quasi polynomial time approximation scheme (QPTAS) for maximum clique in disk graphs~\cite{bonnet2017qptas}. The following year, Bonamy {\em et al.} extended the result to unit ball graphs, and gave a randomised EPTAS for both settings~\cite{bonamy2018eptas}. The current state-of-the-art about the complexity of computing a maximum clique in $d$-ball graphs is summarised in Table~\ref{tableComplexity}.

\begin{table}
\centering
\begin{tabular}{l | c | c }
 & unit $d$-ball graphs & general $d$-ball graphs\\
\hline \hline
$d=1$ & polynomial~\cite{gupta1982efficient}  & polynomial~\cite{gupta1982efficient} \\ 
$d=2$ &  polynomial~\cite{clark1990unit} & Unknown but EPTAS~\cite{bonnet2017qptas, bonamy2018eptas}\\
$d=3$ & Unknown but EPTAS~\cite{bonamy2018eptas} & NP-hard~\cite{bonamy2018eptas}\\
$d=4$ & NP-hard~\cite{bonamy2018eptas} & NP-hard~\cite{bonamy2018eptas} 
\end{tabular}
\caption{Complexity of computing a maximum clique in $d$-ball graphs}
\label{tableComplexity}
\end{table}


Bonamy {\em et al.} show that the existence of an EPTAS is implied by the following fact: For any graph $G$ that is a disk graph or a unit ball graph, the disjoint union of two odd cycles is a forbidden induced subgraph in the complement of~$G$. Surprisingly, the proofs for disk graphs on one hand and unit ball graphs on the other hand are not related. Bonamy {\em et al.} ask whether there is a natural explanation of this common property. They say that such an explanation could be to show the existence of a geometric superclass of disk graphs and unit ball graphs, for which there exists an EPTAS for solving maximum clique.

 By looking at Table~\ref{tableComplexity}, a pattern seems to emerge: The complexity of computing a maximum clique in $(d-1)$-ball graphs and unit $d$-ball graphs might be related. We extend the question of Bonamy {\em et al.} and ask for a class of geometric intersection graphs that $1)$ contains all interval graphs and all unit disk graphs, and $2)$ for which maximum clique can be solved in polynomial time. Recall that the complexity of maximum clique in disk graphs is still open. Therefore a second motivation for our question is that showing the existence of polynomial time algorithms for geometric superclasses of unit disk graphs may help to determine the complexity of maximum clique in disk graphs.

We introduce a class $C$ of geometric intersection graphs which contains all interval graphs and all unit disk graphs, and show that maximum clique can be solved efficiently in $C$. Furthermore, the definition of our class generalises to any dimension, i.e. for any fixed $d\geq 2$ we give a class of geometric intersection graphs that contains all $(d-1)$-ball graphs and all unit $d$-ball graphs. We conjecture that for $d=3$, there exists an EPTAS for computing a maximum clique in the corresponding class. It is necessary that these superclasses be defined as classes of geometric intersection graphs. Indeed, it must be if we want to understand better the reason why efficient algorithms exist for both settings. For instance, taking the union of interval graphs and unit disk graphs would not give any insight, since it is a priori not defined by intersection graphs of some geometric objects.

In order to define the class, we first introduce the concept of \emph{$d$-pancakes}. A $2$-pancake is defined as the union of all unit disks whose centres lie on a line segment $s$, with $s$ itself lying on the $x$-axis. An example is depicted in Figure~\ref{fig:2pancake}. This definition is equivalent to the Minkowski sum of a unit disk centred at the origin and a line segment on the $x$-axis, where the Minkowski sum of two sets $A,B$ is defined as the set $\{a+b\mid a\in A, b  \in B\}$. Similarly a $3$-pancake is the union of all unit balls whose centres lie on a disk $\D$, with $\D$ lying on the $xy$-plane. More generally, we have:

\input{dimension2/2pancake.tex}

 \begin{definition}\normalfont
 A \emph{$d$-pancake} is a $d$-dimensional geometric object. Let us denote by $\{ \xi_1, \xi_2, \dots, \xi_d \}$ the canonical basis of $\Re^d$. A $d$-pancake is defined as the Minkowski sum of the unit $d$-ball centred at the origin and a $(d-1)$-ball in the hyperspace induced by $\{ \xi_1, \xi_2, \dots, \xi_{d-1} \}$.
 \end{definition}
 
 We denote by $\Pi^d$ the class of intersection graphs of some finite collection of $d$-pancakes and unit $d$-balls. In this paper, we give a polynomial time algorithm for solving maximum clique in $\Pi^2$: the intersection graphs class of unit disks and $2$-pancakes. This is to put in contrast with the fact that computing a maximum clique in intersection graphs of unit disks and axis-parallel rectangles (instead of $2$-pancakes) is NP-hard and even APX-hard, as shown together with Bonnet and Miltzow~\cite{bonnet2020maximum}, even though maximum clique can be solved in polynomial time in axis-parallel rectangle graphs~\cite{imai1983finding}.

 Relatedly, it would be interesting to generalise the existence of an EPTAS for maximum clique to superclasses of disk graphs. This was achieved with Bonnet and Miltzow for intersection graphs of homothets of a fixed bounded centrally symmetric convex set~\cite{bonnet2020maximum}. In this paper, we aim at generalising further to intersection graphs of convex pseudo-disks, for which we conjecture the existence of an EPTAS, and give partial results towards proving it. The proof of these partial results relies on geometric permutations of line transversals. We do a case analysis on the existence of certain geometric permutations, and show that some convex pseudo-disks must intersect. Holmsen and Wenger have written a survey on geometric transversals~\cite{wenger2017helly}. The results that are related to line transversals are either of Hadwiger-type, concerned with the conditions of existence of line transversals, or about the maximum number of geometric permutations of line transversals. To the best of our knowledge, we do not know of any result that uses geometric permutations of line transversals to show something else. We consider this tool, together with the polynomial time algorithm for computing a maximum clique in $\Pi^2$, to be our main contributions.

\section{Preliminaries}

\subsection{Graph notations}
Let $G$ be a simple graph. We say that two vertices are \emph{adjacent} if there is an edge between them, otherwise they are \emph{independent}. For a vertex $v$, the set $\mathcal{N}(v)$ denotes its \emph{neighbourhood}, i.e. the set of vertices adjacent to~$v$. We denote by $\omega(G)$, $\alpha(G)$, and $\chi(G)$ the clique number, the independence number and the chromatic number of $G$, respectively.

 We denote by $V(G)$ the vertex set of $G$. Let $H$ be a subgraph of $G$. We denote by $G \setminus H$ the subgraph induced by $V(G) \setminus V(H)$. We denote by $\overline{G}$ the \emph{complement} of $G$, which is the graph with the same vertex set, but where edges and non-edges are interchanged. A \emph{bipartite graph} is graph whose vertex set can be partitioned into two independent sets. A graph is \emph{cobipartite} if its complement is a bipartite graph.

We denote by $\text{iocp}(G)$ the \emph{induced odd cycle packing number} of $G$, i.e. the maximum number of vertex-disjoint induced odd cycles (for each cycle the only edges are the ones making the cycle), such that there is no edge between two vertices of different cycles.

\subsection{Geometric notations}
Throughout the paper we only consider Euclidean spaces with the Euclidean distance. Let $p$ and $p'$ be two points in $\Re^d$. We denote by $(p,p')$ the line going through them, and by $[p,p']$ the line segment with endpoints $p$ and $p'$. We denote by $d(p,p')$ the distance between $p$ and $p'$. For any fixed $d$, we denote by $O$ the origin in $\Re^d$. When $d=2$, we denote by $Ox$ and $Oy$ the $x$ and $y$-axis, respectively. For $d=3$, we denote by $xOy$ the $xy$-plane. We usually denote a $d$-pancake by $P^d$. As a reminder, a $2$-pancake is the Minkowski sum of the unit disk centred at the origin $O$ and a line segment lying on the axis $Ox$.

\begin{definition}\normalfont
Let $\{S_i\}_{1\leq i\leq n}$ be a family of subsets of $\Re^d$. We denote the \emph{intersection graph} of $\{S_i\}$ by $G(\{ S_i\})$. It is the graph whose vertex set is $\{ S_i \mid 1 \leq i \leq n \}$ and where there is an edge between two vertices if and only if the corresponding sets intersect.
\end{definition}

\begin{definition}\normalfont
In $\Re^2$ we denote by $\D(c,\rho)$ a closed disk centred at $c$ with radius $\rho$. Let $\D=\D(c,\rho)$ and $\D'=\D(c',\rho')$ be two intersecting disks. We call \emph{lens induced by $\D$ and $\D'$} the region $\D \cap \D'$. We call \emph{half-lenses} the two closed regions obtained by dividing the lens along the line $(c,c')$.
\end{definition}

For any $x_1\leq x_2$, we denote by $P^2(x_1,x_2)$ the $2$-pancake that is the Minkowski sum of the unit disk centred at $O$ and the line segment with endpoints $x_1$ and $x_2$. Therefore we have $P^2(x_1,x_2) =\bigcup_{x_1 \leq x' \leq x_2} \D((x',0),1)$. Behind the definition of the $d$-pancakes is the idea that they should be the most similar possible to unit $d$-balls. In particular $2$-pancakes should behave as much as possible like unit disks. This is perfectly illustrated when the intersection of a $2$-pancake and a unit disk is a lens, as the intersection of two unit disks would be.

\begin{definition}\normalfont
 Let $\{P^2_j\}_{1 \leq j \leq n}$ be a set of $2$-pancakes. For any unit disk $\D$, we denote by $L(\D,\{P^2_j\})$, or simply by $L(\D)$ when there is no risk of confusion, the set of $2$-pancakes in $\{P^2_j\}$ whose intersection with $\D$ is a lens. 
\end{definition}

Let $\D$ denote $\D(c,1)$ for some point $c$. Observe that if a $2$-pancake $P^2(x_1,x_2)$ for some $x_1 \leq x_2$ is in $L(\D)$, then the intersection between $\D$ and $P^2(x_1,x_2)$ is equal to $\D \cap \D((x_1,0),1)$ or $\D \cap \D((x_2,0),1)$. We make an abuse of notation and denote by $d(\D,P^2(x_1,x_2))$ the smallest distance between $c$ and a point in the line segment $[x_1,x_2]$. Observe that if the intersection between $\D$ and $P^2(x_1,x_2)$ is equal to $\D \cap \D((x_1,0),1)$, then $d(\D,P^2(x_1,x_2))=d(c,(x_1,0))$, and otherwise $d(\D,P^2(x_1,x_2))=d(c,(x_2,0))$. The following observation gives a characterisation of when the intersection between a unit disk and a $2$-pancake is a lens.

\begin{observation}\normalfont
\label{obs:lens}
Let $\D((c_x,c_y),1)$ be a unit disk intersecting with a $2$-pancake $P^2(x_1,x_2)$. Then their intersection is a lens if and only if ($c_x \leq x_1$ or $c_x \geq x_2)$ and the interior of $\D((c_x,c_y),1)$ does not contain any point in $\{(x_1, \pm 1),(x_2, \pm 1)\}$.
\end{observation}

The observation follows immediately from the fact that the intersection is a lens if and only if $\D((c_x,c_y),1)$ does not contain a point in the open line segment between the points $(x_1,-1)$ and $(x_2,-1)$, nor in the open line segment between the points $(x_1,1)$ and $(x_2,1)$.

\section{Results}\label{section:results}


 We answer in Section~\ref{section:Pi2} the $2$-dimensional version of the question asked by Bonamy {\em et al.}~\cite{bonamy2018eptas}: We present a polynomial time algorithm for computing a maximum clique in a geometric superclass of interval graphs and unit disk graphs.
 
  \begin{theorem}\label{thm:polyPi2}
 There exists a polynomial time algorithm for computing a maximum clique in $\Pi^2$, even without a representation.
 \end{theorem}
 
Kang and M{\"u}ller have shown that for any fixed $d\geq 2$, the recognition of unit $d$-ball graphs is NP-hard, and even $\exists \Re$-hard~\cite{kang2012sphere}. We conjecture that it is also hard to test whether a graph is in $\Pi^d$ for any fixed $d\geq 3$, and prove it for $d=2$. 
 
\begin{theorem}\label{thm:recognition}
Testing whether a graph is in $\Pi^2$ is NP-hard, and even $\exists \Re$-hard.
\end{theorem}
 
 The proof of Theorem~\ref{thm:recognition} figures in Section~\ref{section:recognition}. It immediately implies that given a graph $G$ in $\Pi^2$, finding a representation of $G$ with $2$-pancakes and unit disks is NP-hard. Therefore having a robust algorithm as defined in~\cite{raghavan2001robust} is of interest. The algorithm of Theorem~\ref{thm:polyPi2} takes any abstract graph as input, and outputs a maximum clique or a certificate that the graph is not in $\Pi^2$. 
 
 Our polynomial time algorithm for maximum clique in $\Pi^2$ gives some insight in understanding why the complexity of maximum clique in disk graphs is still unknown. The class of interval graphs is arguably small: there is no induced cycle of length at least $4$. Likewise, one can say that the class of unit disk graphs is small, as there is no star with at least $6$ leaves. However, one can realise with disks arbitrarily large induced cycles and stars. One could have wondered whether when looking for a geometric class of graphs, wanting both arbitrarily large induced cycles and stars would force too much complexity. Our results with $\Pi^2$ shows that actually, this is not where the difficulty lies. Indeed, one can realise with unit disks and $2$-pancakes arbitrarily large induced cycles and stars. To solve maximum clique in disk graphs, or to show NP-hardness, it seems a good idea to investigate what are the disk graphs which are not in $\Pi^2$.

 Concerning $\Pi^3$, we conjecture the following:
 \begin{conjecture}\label{conj:iocp}
There exists an integer $K$ such that for any graph $G$ in $\Pi^3$, we have $\text{iocp}(\overline{G})\leq K$.
\end{conjecture}

We show in Section~\ref{sec:conjToEPTAS} that this would be sufficient to obtain an EPTAS.
 \begin{theorem}\label{Thm:conjImpliesEptas}
If Conjecture~\ref{conj:iocp} holds, there exists a randomised EPTAS for computing a maximum clique in $\Pi^3$, even without a representation.
 \end{theorem}

 By construction the class $\Pi^d$ contains all $(d-1)$-ball graphs and all unit $d$-ball graphs. Indeed a $(d-1)$-ball graph can be realised by replacing in a representation each $(d-1)$-ball by a $d$-pancake. In addition to this property, we want fast algorithms for maximum clique in $\Pi^d$. The definition of $\Pi^d$ may seem unnecessarily complicated. 
 The most surprising part of the definition is probably the fact that we use $d$-pancakes instead of simply using $(d-1)$-balls restricted to be in the same hyperspace of $\Re^d$. However, we show in Section~\ref{section:motivation} that our arguments for proving fast algorithms would not hold with such a definition.

We give partial results toward showing the existence of an EPTAS for maximum clique in intersection graphs of \emph{convex pseudo-disks}. We say that a graph is a \emph{convex pseudo-disk graph} if it is the intersection graph of convex sets in the plane such that the boundaries of every pair intersect at most twice. We denote by $\mathcal{G}$ the class of intersection graphs of convex pseudo-disks. A structural property used to show the existence of an EPTAS for disk graphs is that for any disk graph $G$, $\text{iocp}(\overline{G})\leq 1$. 
The proof of Bonnet {\em et al.} relies heavily on the fact that disks have centres~\cite{bonnet2017qptas}. However, convex pseudo-disks do not, therefore adapting the proof in this new setting does not seem easy. While we were not able to extend this structural result to the class $\mathcal{G}$, we show a weaker property: The complement of a triangle and an odd cycle is a forbidden induced subgraph in~$\mathcal{G}$. We write ``complement of a triangle" to make the connection with \emph{iocp} clear, but note that actually the complement of a triangle is an independent set of three vertices. Below we state this property more explicitly.

\begin{theorem}\label{thm:convexPseudoDisks}
Let $G$ be in $\mathcal{G}$. If there exists an independent set of size $3$, denoted by $H$, in $G$, and if for any $u\in H$ and $v \in G \setminus H$, the edge $\{u,v\}$ is an edge of $G$, then $G \setminus H$ is cobipartite.
\end{theorem}

Note that a cobipartite graph is not the complement of an odd cycle. Given the three pairwise non-intersecting convex pseudo-disks in $H$, we give a geometric characterisation of the two independent sets in the complement of $G \setminus H$. We conjecture that Theorem~\ref{thm:convexPseudoDisks} is true even when $H$ is the complement of any odd cycle, which implies:

\begin{conjecture}\label{conj:convexPseudo}
For any convex pseudo-disk graph $G$, we have $\text{iocp}(\overline{G})\leq 1$.
\end{conjecture}

If Conjecture~\ref{conj:convexPseudo} holds, it is straightforward to obtain an EPTAS for maximum clique in convex pseudo-disks graphs, by using the method of Bonamy {\em et al.}~\cite{bonamy2018eptas}.

\section{Computing a maximum clique in $\Pi^2$ in polynomial time}
\label{section:Pi2}
 In this section we prove Theorem~\ref{thm:polyPi2}. We first give a proof when a representation is given. The idea of the algorithm is similar to the one of Clark, Colbourn and Johnson~\cite{clark1990unit}. We prove that if $u$ and $v$ are the most distant vertices in a maximum clique, then $\mathcal{N}(u) \cap \mathcal{N}(v)$ is cobipartite. In a second part, we give a robust algorithm, meaning that it does not require a representation, using tools introduced by Raghavan and Spinrad~\cite{raghavan2001robust}.
 
 \subsection{Computing a maximum clique with a representation}

In their proof, Clark, Colbourn and Johnson use the following fact: if $c$ and $c'$ are two points at distance~$\rho$, then the diameter of the half-lenses induced by $\D(c,\rho)$ and $\D(c',\rho)$ is equal to $\rho$. We prove here a similar result.

\begin{lemma}\label{lemma:diameterHalfLenses}
Let $c$ and $c'$ be two points at distance $\rho$, and let be $\rho' \geq \rho$. Then the diameter of the half-lenses induced by $\D(c,\rho)$ and $\D(c',\rho')$ is at most $\rho'$.
\end{lemma}

\begin{proof}
First note that if $\rho'> 2\rho$ then the half-lenses are half-disks of $\D(c,\rho)$. The diameter of these half-disks is equal to $2\rho$, which is smaller than $\rho'$. Let us now assume that we have $\rho' \leq 2 \rho$. The boundary of the lens induced by $\D(c,\rho)$ and $\D(c',\rho')$ consists of two arcs. The line $(c,c')$ intersects exactly once with each arc. One of these two intersections is $c'$, we denote by $c''$ the other. Let us consider the disk $\D(c'',\rho')$. Note that it contains the disk $\D(c,\rho)$. Therefore the lens induced by $\D(c,\rho)$ and $\D(c',\rho')$ is contained in the lens induced by $\D(c'',\rho')$ and $\D(c',\rho')$, whose half-lenses have diameter $\rho'$. The claim follows from the fact that the half-lenses of the first lens are contained in the ones of the second lens.
\end{proof}

Before stating the next lemma, we introduce the following definition:
\begin{definition}\normalfont
Let $\{S_i\}_{1\leq i\leq n}$ and $\{S'_j\}_{1\leq j\leq n'}$ be two families of sets in $\Re^2$. We say that $\{S_i\}$ and $\{S'_j\}$ \emph{fully intersect} if for all $1\leq i \leq n$ and $1\leq j \leq n'$ the intersection between $S_i$ and $S'_j$ is not empty.
\end{definition}

\begin{lemma}\label{lemma:1disk1pancake}
Let $\D:=\D(c,1)$ be a unit disk and let $P^2:=P^2(x_1,x_2)$ be in $L(\D)$. Let $\{\D_i\}$ be a set of unit disks that fully intersect with $\{\D,P^2\}$, such that for any $\D_i$ we have $d(\D,\D_i) \leq d(\D,P^2)$. Moreover if $P^2$ is in $L(\D_i)$ we require $d(\D_i,P^2) \leq d(\D,P^2)$. Also let $\{P^2_j\}$ be a set of $2$-pancakes that fully intersect with $\{\D,P^2\}$, such that for any $P^2_j$ in $\{P^2_j\} \cap L(D)$, we have $d(\D,P^2_j) \leq  d(\D,P^2)$. Then $G(\{\D_i\} \cup \{P^2_j\})$ is cobipartite.
\end{lemma}

\begin{proof}

The proof is illustrated in Figure~\ref{1disk1pancake}. Without loss of generality, let us assume that the intersection between $\D$ and $P^2$ is equal to $\D \cap \D((x_1,0),1)$. Remember that by definition we have $x_1 \leq x_2$. Let $P^2(x'_1,x'_2)$ be a $2$-pancake in $\{P^2_j\}$. As it is intersecting with $P^2$, we have $x'_2\geq x_1-2$. Assume by contradiction that we have $x'_1 > x_1$. Then with Observation~\ref{obs:lens}, we have that $P^2(x'_1,x'_2)$ is in $L(\D)$ and $d(\D,P^2(x'_1,x'_2))> d(\D,P^2)$, which is impossible. Therefore we have $x'_1 \leq x_1$, and so $P^2(x'_1,x'_2)$ must contain $\D((x',0),1)$ for some $x'$ satisfying $x_1-2 \leq x' \leq x_1$. As the line segment $[(x_1-2,0),(x_1,0)]$ has length $2$, the $2$-pancakes in $\{P^2_j\}$ pairwise intersect.

We denote by $\rho$ the distance $d(\D,P^2)$. Let $\D(c_i,1)$ be a unit disk in $\{\D_i\}$. By assumption, $c_i$ is in $\D(c,\rho) \cap \D((x_1,0),2)$. We then denote by $R$ the lens that is induced by $\D(c,\rho)$ and $\D((x_1,0),2)$. We cut the lens into two parts with the line $(c,(x_1,0))$, and denote by $R_1$ the half-lens that is not below this line, and by $R_2$ the half-lens that is not above it. With Lemma~\ref{lemma:diameterHalfLenses}, we obtain that the diameter of $R_1$ and $R_2$ is at most $2$. Let us assume without loss of generality that $c$ is not below $Ox$. We denote by $X_1$ the set of unit disks in $\{\D_i\}$ whose centre is in $R_1$. We denote by $X_2$ the union of $\{P^2_j\}$ and of the set of unit disks in $\{\D_i\}$ whose centre is in $R_2$. Since the diameter of $R_1$ is $2$, any pair of unit disks in $X_1$ intersect, therefore $G(X_1)$ is a complete graph. To show that $G(X_2)$ is a complete graph too, it remains to show that any unit disk $\D(c_i,1)$ in $X_2$ and any $2$-pancake $P^2(x'_1,x'_2)$ in $\{P^2_j\}$ intersect. We denote by $P^2_+$ the following convex shape: $\cup_{x'_1\leq x \leq x'_2}\D((x,0),2)$. Note that the fact that $\D(c_i,1)$ and $P^2(x'_1,x'_2)$ intersect is equivalent to having $c_i$ in $P^2_+$. Let us consider the horizontal line going through $c$, and let us denote by $c'$ the left intersection with the circle centred at $(x_1,0)$ with radius $2$. We also denote by $r_2$ the extremity of $R$ that is in $R_2$.

Let us assume by contradiction that $c_i$ is above the line segment $[c,c']$. As by assumption $c_i$ is in $R_2$, it implies that the $x$-coordinate of $c_i$ is smaller than the one of $c$. Therefore $P^2$ is in $L(\D_i)$ and $d(\D_i,P^2)>d(\D,P^2)$, which is impossible by assumption. Let us denote by $R_{2,-}$ the subset of $R_2$ that is not above the line segment $[c,c']$. To prove that $\D(c_i,1)$ and $P^2(x'_1,x'_2)$ intersect, it suffices to show that $P^2_+$ contains $R_{2,-}$. As shown above, $P^2(x'_1,x'_2)$ contains  $\D((x',0),1)$ for some $x'$ satisfying $x_1-2 \leq x' \leq x_1$. This implies that $P^2_+$ contains $\D((x_1-2,0),2) \cap \D((x_1,0),2)$, and in particular contains $(x_1,0)$. Moreover as $c$ is not below $Ox$, $r_2$ is also in $\D((x_1-2,0),2) \cap \D((x_1,0),2)$. As $P^2$ intersects $\D$, $P^2_+$ contains $c$. Let us assume by contradiction that $P^2_+$ does not contain $c'$. Then $x'_2$ must be smaller than the $x$-coordinate of $c'$, because otherwise the distance $d((x'_2,0),c')$ would be at most $d((x_1,0),c')$, which is equal to $2$. But then if $P^2_+$ does not contain $c'$, then it does not contain $c$ either, which is a contradiction. We have proved that $P^2_+$ contains the points $(x_1,0)$, $c$, $c'$ and $r_2$. By convexity, and using the fact that two circles intersect at most twice, we obtain that $R_{2,-}$ is contained in $P^2_+$. This shows that any two elements in $X_2$ intersect, which implies that $G(X_2)$ is a complete graph. Finally, as $X_1 \cup X_2 = \{\D_i\} \cup \{P^2_j\}$, we obtain that $G(\{\D_i\} \cup \{P^2_j\})$ can be partitioned into two cliques, i.e. it is cobipartite.
\end{proof}


\begin{figure}
    \centering
    \input{dimension2/1disk1pancake.tikz}
    \caption{Illustration of the proof of Lemma~\ref{lemma:1disk1pancake}}
    \label{1disk1pancake}
\end{figure}

\begin{lemma}\label{lemma:help2disks}
Let $\D:=\D((c_x,c_y),1)$ and $\D':=\D((c'_x,c'_y),1)$ be two unit disks such that $c_x \leq c_x'$. Let $P^2_1:=P^2(x_1,x_2)$ be a $2$-pancake intersecting with $\D$ and $\D'$, such that $ x_1 \geq c_x$ and $P^2_1$ is not in $L(\D)$. If $P^2_2:=P^2(x'_1,x'_2)$ is a $2$-pancake intersecting with $\D$ and $\D'$, but not intersecting with $P^2_1$, then $P^2_2$ is in $L(\D) \cap L(\D')$.
\end{lemma}

\begin{proof}
The proof is illustrated in Figure~\ref{help2disks}. First let us prove that $P^2_2$ cannot be on the right side of $P^2_1$, i.e. we have $x'_1\leq x_1$. Let us assume by contradiction $x'_1> x_1$. As $P^2_1$ and $P^2_2$ are not intersecting, we have $x'_1> x_1+2$. Hence, since we assume $x_1 \geq c_x$, we obtain $d(c,(x'_1,0))>2$, which is impossible. Therefore we have $x'_1\leq x_1$, and even $x'_1< x_1-2$ since $P^2_1$ and $P^2_2$ are not intersecting.

 Without loss of generality, let us assume $c_y\geq0$. Let us consider the horizontal line $\ell$ with height $1$. By assumption it intersects with the circle centred at $(c_x,c_y)$ with unit radius. There are at most two intersections, and we denote by $x_\ell$ the $x$-coordinate of the one to the right. As $P^2_1$ is not in $L(\D)$, we have $x_1\leq x_\ell$. Then, since $x'_1< x_1-2$ and the fact that $\D$ has diameter $2$, we know that the points $(x'_1,1)$ and $(x'_1,-1)$ are not in $\D$, which implies that $P^2_2$ is in $L(\D)$. Likewise as we have $c_x \leq c_x'$, the points $(x'_1,1)$ and $(x'_1,-1)$ are not in $\D'$, and so $P^2_2$ is in $L(\D) \cap L(\D')$.
\end{proof}

\input{dimension2/help2disks.tex}

\begin{lemma}\label{lemma:2disks}
Let $\D:=\D(c,1)$ and $\D':=\D(c',1)$ be two intersecting unit disks. Let $\{\D_i\}$ be a set of unit disks that fully intersect with $\{\D,\D'\}$, such that for each unit disk $\D_i$ we have $d(\D,\D_i)\leq d(\D,\D')$ and $d(\D',\D_i)\leq d(\D,\D')$. Also let $\{P^2_j\}$ be a set of $2$-pancakes that fully intersect with $\{\D,\D'\}$, such that for any $P^2_j$ in $\{P^2_j\} \cap L(D)$, we have $d(\D,P^2_j) \leq  d(\D,\D')$, and for any $P^2_j$ in $\{P^2_j\} \cap L(D')$, we have $d(\D',P^2_j) \leq  d(\D,\D')$. Then $G(\{\D_i\} \cup \{P^2_j\})$ is cobipartite.
\end{lemma}

\begin{proof}
We denote by $\rho$ the distance between $c$ and $c'$. We denote by $R$ the lens induced by $\D(c,\rho)$ and $\D(c',\rho)$. We cut $R$ with the line segment $[c,c']$, which partitions $R$ into two half-lenses that we denote by $R_1$ and $R_2$. By assumption, the centre of any unit disk in $\{\D_i\}$ must be in $R$. Since $R_1$ and $R_2$ have diameter $\rho$ which is at most $2$, any two unit disks having their centres in $R_1$ must intersect, and the same holds with $R_2$. Therefore $G(\{\D_i\})$ is cobipartite, which is the claim if $\{P^2_j\}$ is empty.

We now assume that $\{P^2_j\}$ is not empty. In order to show the claim, we do a case analysis according to whether the intersection between $\D(c,2) \cap \D(c',2)$ and $Ox$ is empty or not. Let us assume that the latter holds, as shown in Figure~\ref{Fig:nonEmptyIntersection}. Let $P^2$ be in a $2$-pancake in $\{P^2_j\}$. As $P^2$ intersects with $\D$, $P^2$ contains a unit disk that intersects with $\D$. Likewise, $P^2$ contains a unit disk that intersects with $\D'$. This implies that $P^2$ contains a point in $\D(c,2) \cap Ox$ and a point in $\D(c',2) \cap Ox$. By convexity of a $2$-pancake, $P^2$ contains a point $(x',0)$, where $(x',0)$ is in $\D(c,2) \cap \D(c',2)\cap Ox$. We denote by $R^+$ the lens that is induced by $\D(c,2)$ and $\D(c',2)$ and cut it with the line $(c,c')$. We denote by $R^+_1$ (respectively $R^+_2)$ the half-lens that contains $R_1$ (respectively $R_2$). Let us assume that $(x',0)$ is in $R^+_1$. By assumption $\D((x',0),2)$ contains $c$ and $c'$. Let us consider the third extremity of $R_1$, along with $c$ and $c'$, that we denote by $r_1$. By making a circle centred at $r_1$ grow, we observe that the farthest point from $r_1$ in $R^+_1$ can only be at one of the three extremities of $R^+_1$. However by Lemma~\ref{lemma:diameterHalfLenses} these distances are at most $2$, which implies that the distance between $(x',0)$ and $r_1$ is at most $2$. Using the fact that two circles intersect at most twice, we obtain that $R_1$ is contained in $\D((x',0),2)$. Therefore $P^2$ intersect with all unit disks whose centre is in $R_1$, and with all $2$-pancakes in $L(\D) \cap L(\D')$ that contain a disk whose centre is in $R_1$.
Let $P^2(x_1,x_2)$ and $P^2(x'_1,x'_2)$ be two $2$-pancakes in $\{P^2_j\}$ such that they contain each a unit disk whose centre is in $R^+_1$, but such that they do not contain a unit disk whose centre is in $R_1$. In particular, $P^2(x_1,x_2)$ and $P^2(x'_1,x'_2)$ are not in $L(\D) \cap L(\D')$. We claim that they intersect. Suppose by contradiction that they do not. Without loss of generality, let us assume that $P^2(x_1,x_2)$ is to the right of $P^2(x'_1,x'_2)$, and that $c_x \leq c'_x$, where $c_x$ and $c'_x$ denote the $x$-coordinate of $c$ and $c'$ respectively. Since $P^2(x_1,x_2)$ does not contain a disk in $R_1$, and since it is on the right side of $P^2(x'_1,x'_2)$, it implies that it does not contain a disk with centre in $\D(c,\rho)$. Therefore $P^2(x_1,x_2)$ cannot be in $L(\D)$. Moreover the fact that it does not contain a disk with centre in $\D(c,\rho)$ implies $x_1\geq c_x$. We finally apply Lemma~\ref{lemma:help2disks} to obtain a contradiction. We denote $X_1$ the set of unit disks whose centre is in $R_1$ and $2$-pancakes that contain a disk whose centre is in $R^+_1$. We know that two unit disks in $X_1$ intersect. Moreover we have shown that a $2$-pancake and a unit disk in $X_1$ intersect. For a pair of two pancakes, if one of them contains a disk whose centre is in $R_1$ it is done for the same reasons. If none of them does, then we have shown above that they intersect. This shows that $G(X_1)$ is a complete graph. By defining $X_2$ as the set of the remaining disks and $2$-pancakes, using the symmetry of the problem we obtain that $G(X_2)$ is also a complete graph.

Now let us assume that the intersection between $\D(c,2) \cap \D(c',2)$ and $Ox$ is empty, as shown in Figure~\ref{emptyIntersection}. As $\{P^2_j\}$ is not empty, the set $Ox \setminus (\D(c,2) \cup \D(c',2))$ consists of three connected component, one of them bounded. We denote by $s$ the closed line segment consisting of the bounded connected component and its boundaries. Any $2$-pancake $P^2$ in $\{P^2_j\}$ contains a point in $\D(c,2) \cap Ox$, otherwise $P^2$ would not intersect with $\D$. Likewise $P^2$ contains a point in $\D(c',2) \cap Ox$, and therefore contains $s$. This implies that all $2$-pancakes in $\{P^2_j\}$ pairwise intersect. Let us assume without loss of generality that $R_1$ is closer to $Ox$ than $R_2$. Let us show that any $2$-pancake $P^2$ in $\{P^2_j\}$ and any unit disk whose centre is in $R_1$ intersect. This is equivalent to show that for any point $p$ in $R_1$, there is a point in $P^2 \cap Ox$ that is at Euclidean distance at most $2$ from $p$. We denote by $P^2_+$ the Minkowski sum of the disk with radius $2$ centred at $O$ and the line segment $s$, i.e. $P^2_+=\cup_{x' \in s}\D(x',2)$. Note that $P^2_+$ is convex. We claim that $P^2_+$ contains $R_1$, which implies the desired property. Since $s$ contains a point $p_1$ in $\D(c,2)$, we know that $P^2_+$ contains $c$. Likewise, as $s$ contains a point $p_2$ in $\D(c',2)$, then $P^2_+$ contains $c'$, and therefore by convexity the whole line segment $[c,c']$. Therefore $P^2_+$ contains the quadrilateral $cc'p_2p_1$. If this quadrilateral contains $R_1$ we are done. Otherwise, it may not contain a circular segment of the disk $\D(c',\rho)$ or a circular segment of the disk $\D(c,\rho)$. Let us assume that we have the worst case, meaning that both circular segments are not in $cc'p_2p_1$. Let us consider the circle $\C_1$ centred at $p_1$ with radius $2$, and the circle $\C'$ centred at $c'$ with radius $\rho$. The two circles intersect at $c$. Let us consider the point $p'_1$ that is at the intersection between $\C'$ and the line segment $[c,p_1]$. By definition, $p'_1$ is inside the disk induced by $\C_1$. As two circles intersect at most twice, we obtain that the arc $cp'_1$ centred at $c'$ with radius $\rho$ is contained in the disk induced by $\C_1$, and therefore also in $P^2_+$. By convexity, we know that the circular segment of the disk $\D(c',\rho)$ with chord $[c,p'_1]$ is in $P^2_+$. We can apply the same arguments for the other side to show that $R_1$ is in $P^2_+$. Hence by defining $X_1$ as the set of disks whose centre centre is in $R_1$, union the set of $2$-pancakes, and $X_2$ as the set of disks whose centre is in $R_2$, we have that $G(X_1)$ and $G(X_2)$ are complete graphs.
\end{proof}


\begin{figure}
    \centering
    \input{dimension2/nonEmptyIntersection.tikz}
    \caption{First case: $\D(c,2) \cap \D(c',2) \cap Ox \neq \emptyset$}
    \label{Fig:nonEmptyIntersection}
\end{figure}

\begin{figure}[ht!]
    \centering

    \input{dimension2/emptyIntersection.tikz}
    \caption{Second case: $\D(c,2) \cap \D(c',2) \cap Ox= \emptyset$}
    \label{emptyIntersection}

\end{figure}

Note that Lemma~\ref{lemma:1disk1pancake} and Lemma~\ref{lemma:2disks} give a polynomial time algorithm for maximum clique in $\Pi^2$ when a representation is given. First compute a maximum clique that contains only $2$-pancakes, which can be done in polynomial time since the intersection graph of a set of $2$-pancakes is an interval graph~\cite{gupta1982efficient}. Then for each unit disk $\D$, compute a maximum clique which contains exactly one unit disk, $\D$, and an arbitrary number of $2$-pancakes. Because finding out whether a unit disk and a $2$-pancake intersect takes constant time, computing such a maximum clique can be done in polynomial time. Note that if a maximum clique contains at least two unit disks, then in quadratic time we can find in this maximum clique either a pair of unit disks or a unit disk and a $2$-pancake whose intersection is a lens, such that the conditions of Lemma~\ref{lemma:1disk1pancake} or of Lemma~\ref{lemma:2disks} are satisfied. By applying the corresponding lemma, we know that we are computing a maximum clique in a cobipartite graph, which is the same as computing a maximum independent set in a bipartite graph. As this can be done in polynomial time~\cite{edmonds1972theoretical}, we can compute a maximum clique in $\Pi^2$ in polynomial time when the representation is given.

\subsection{Computing a maximum clique without a representation}

To obtain an algorithm that does not require a representation, we use the notion of cobipartite neighbourhood edge elimination ordering (CNEEO) as introduced by Raghavan and Spinrad~\cite{raghavan2001robust}. Let $G$ be a graph with $m$ edges. Let $\Lambda=e_1,e_2, \dots, e_m$ be an ordering of the edges. Let $G_\Lambda(k)$ be the subgraph of $G$ with edge set $\{e_k,e_{k+1}, \dots, e_m\}$. For each $e_k=(u,v)$, $N_{\Lambda,k}$ is defined as the set of vertices adjacent to $u$ and $v$ in $G_\Lambda(k)$.

\begin{definition}[Raghavan and Spinrad~\cite{raghavan2001robust}]\normalfont
We say that an edge ordering $\Lambda=\{e_1,e_2,\dots,e_m\}$ is a CNEEO if for each $e_k$, $N_{\Lambda,k}$ induces a cobipartite graph in $G$.
\end{definition}

\begin{lemma}[Raghavan and Spinrad~\cite{raghavan2001robust}]
\label{lemma:CNEEOpolytime}
Given a graph $G$ and a CNEEO $\Lambda$ for $G$, a maximum clique in $G$ can be found in polynomial time.
\end{lemma}

They propose a greedy algorithm for finding a CNEEO: When having chosen the first $i-1$ edges $e_1, \dots, e_{i-1}$, try every remaining edge one by one until finding one that satisfies the required property. If no such edge exists, return that the graph does not admit a CNEEO, which follows from Lemma~\ref{lemma:greedyCNEEO}.

\begin{lemma}[Raghavan and Spinrad~\cite{raghavan2001robust}]
\label{lemma:greedyCNEEO}
If $G$ admits a CNEEO, then the greedy algorithm finds a CNEEO for $G$.
\end{lemma}

To show that it is possible to compute a maximum clique in a graph $G$ in $\Pi^2$, we show that such a graph admits a CNEEO. As noted by Raghavan and Spinrad, the algorithm computes a maximum clique for any graph that admits a CNEEO, and otherwise states that the given graph does not admit a CNEEO. In particular, the algorithm does not say whether the graph is indeed in $\Pi^2$, and cannot be used for recognition.

\begin{theorem}\label{thm:CNEEO}
If a graph $G$ is in $\Pi^2$, then $G$ admits a CNEEO.
\end{theorem}

Theorem~\ref{thm:CNEEO}, Lemma~\ref{lemma:CNEEOpolytime} and Lemma~\ref{lemma:greedyCNEEO} immediately imply Theorem~\ref{thm:polyPi2}. To prove Theorem~\ref{thm:CNEEO}, we use two more lemmas. 

\begin{lemma}\label{lemma:1disk}
Let $\D=\D((c_x,c_y),1)$ be a unit disk. Let $\{P^2_j\}$ be a set of $2$-pancakes that all intersect with $\D$. Then $G(\{P^2_j\})$ is cobipartite.
\end{lemma}

\begin{proof}
Let $P^2(x_1,x_2)$ be in $\{P^2_j\}$. By triangular inequality we have $x_1 \leq c_x+2$ or $x_2 \geq c_x-2$. It implies that $P^2(x_1,x_2)$ contains the line segment $[(x'-1,0),(x'+1,0)]$ for some $x'$ satisfying $c_x-2 \leq x' \leq c_x+2$. We define $X_1$ as the set of $2$-pancakes in $\{P^2_j\}$ that contain the line segment $[(x'-1,0),(x'+1,0)]$ for some $x'$ satisfying $c_x-2 \leq x' \leq c_x$, and $X_2$ as $\{P^2_j\} \setminus X_1$. We obtain that $G(X_1)$ and $G(X_2)$ are complete graphs.
\end{proof}

\begin{lemma}\label{lemma:2pancakes}
Let $P^2=P^2(x_1,x_2)$ and $P'^2=P^2(x'_1,x'_2)$ be two intersecting $2$-pancakes. Let $\{P^2_j\}$ be a set of $2$-pancakes that fully intersect with $\{P^2,P'^2\}$, such that for any $P^2_j$ in $\{P^2_j\}$, $P^2_j$ is not contained in $P^2$ nor in $P'^2$. Then $G(\{P^2_j\})$ is cobipartite.
\end{lemma}

\begin{proof}
Let $P^2(x''_1,x''_2)$ be in $\{P^2_j\}$. Let us first assume that one of $P^2,P'^2$ is contained in the other. Without loss of generality, let us assume that $P^2$ is contained in $P'^2$, which is equivalent to having $x'_1\leq x_1 \leq x_2 \leq x'_2$. By assumption, as $P^2(x''_1,x''_2)$ is not contained in $P^2$, we have $x''_1<x_1$ or $x_2<x''_2$. As $P^2(x''_1,x''_2)$ intersects with $P^2$, it implies that $P^2(x''_1,x''_2)$ contains $(x_1-1,0)$ or $(x_2+1,0)$. We define $X_1$ as the set of $2$-pancakes in $\{P^2_j\}$ that contains $(x_1-1,0)$, and $X_2$ as $\{P^2_j\}\setminus X_1$. We obtain that $G(X_1)$ and $G(X_2)$ are complete graphs.

If none of $P^2,P'^2$ is contained in the other, we can assume without loss of generality that $x_1 \leq x'_1 \leq x_2 \leq x'_2$. Therefore we have $x''_1<x'_1$ or $x_2<x''_2$, which implies that $P^2(x''_1,x''_2)$ contains $(x'_1-1,0)$ or $(x_2+1,0)$. We conclude as above.
\end{proof}

\begin{proof}[Proof of Theorem~\ref{thm:CNEEO}]
Let us consider any representation of $G$ with unit disks and $2$-pancakes. We divide the set of edges into three sets: $E_1$, $E_2$ and $E_3$. $E_1$ contains all the edges between a pair of unit disks, or between a unit disk $\D$ and a $2$-pancake in $L(D)$. $E_2$ contains the edges between a unit disk and a $2$-pancake that are not in $E_1$. $E_3$ contains the edges between a pair of $2$-pancakes. For an edge $e=\{u,v\}$ in $E_1$, we call length of $e$ the distance between $u$ and $v$, be they unit disks or a unit disk $\D$ and a $2$-pancake in $L(D)$. We order the edges in $E_1$ by non increasing length, which gives an ordering $\Lambda_1$. We take any ordering $\Lambda_2$ of the edges in $E_2$. For $E_3$, we take any ordering $\Lambda_3$ such that for any edge $e=\{u,v\}$, no edge after $e$ in $\Lambda_3$ contains a $2$-pancake contained in $u$ or $v$. This can be obtained by considering the smallest $2$-pancakes first. We finally define an ordering $\Lambda = \Lambda_1 \Lambda_2 \Lambda_3$ on $E$. Let us consider an edge $e_k$. If $e_k$ is in $E_1$, Lemma~\ref{lemma:1disk1pancake} and Lemma~\ref{lemma:2disks} show that $N_{\Lambda,k}$ induces a cobipartite graph. If $e_k$ is in $E_2$, we use Lemma~\ref{lemma:1disk}, and if $e_k$ is in $E_3$, we conclude with Lemma~\ref{lemma:2pancakes}. This shows that $\Lambda$ is a CNEEO. 
\end{proof}

\subsection{A motivation for $\Pi^d$}
\label{section:motivation}

As we define it, $\Pi^d$ is the class of intersection graphs of $d$-pancakes and unit $d$-balls. The properties that we desire are:
\begin{enumerate}
    \item $\Pi^d$ contains $(d-1)$-ball graphs and unit $d$-ball graphs,
    \item Maximum clique can be computed as fast in $\Pi^d$ as in $(d-1)$-ball graphs and unit $d$-ball graphs.
\end{enumerate}

Let $\{\xi_i\}_{1 \leq i \leq d}$ be the canonical basis of $\Re^d$. Let us consider another class $\tilde{\Pi}^d$, that might a priori satisfy those properties.

\begin{definition}
We denote by $\tilde{\Pi}^d$ the class of intersection graphs of $(d-1)$-balls lying on the hyperspace induced by $\{ \xi_1, \xi_2, \dots, \xi_{d-1} \}$ and of unit $d$-balls.
\end{definition}

This class might look more natural since it makes use only of balls and not of pancakes. It contains by definition $(d-1)$-ball graphs and unit $d$-ball graphs. Moreover, as we want to be able to compute a maximum clique fast, we are looking for a ``small'' superclass. However, while we do not rule out the existence of a polynomial algorithm for computing a maximum clique in $\tilde{\Pi}^2$, we demonstrate that Lemma~\ref{lemma:2disks} does not hold in $\tilde{\Pi}^2$. 

The counterexample is illustrated in Figure~\ref{2disksCounterExample}. We have two intersecting unit disks $\D$ and $\D'$. Moreover each one of $\D_1$, $\D_2$ and the line segment $[x_1,x_2]$ intersects with both $\D$ and $\D'$. The distances $d(\D,\D_1)$, $d(\D',\D_1)$ are smaller than $d(\D,\D')$, and the same hold for $\D_2$. To define $L(\D)$, a natural way would be to use the same characterisation as in Observation~\ref{obs:lens}. Therefore the line segment $[x_1,x_2]$ is not in $L(\D)$ nor in $L(\D')$. However, $G(\{\D_1,\D_2,[x_1,x_2]\})$ is an edgeless graph with three vertices, and therefore is not cobipartite. 
\input{motivation/2disksCounterExample.tex}

\section{Recognition of graphs in $\Pi^2$}
\label{section:recognition}
We show that testing whether a graph can be obtained as the intersection graph of unit disks and $2$-pancakes is hard, as claimed in Theorem~\ref{thm:recognition}. 

\begin{proof}[Proof of Theorem~\ref{thm:recognition}]
We do a reduction from recognition of unit disk graphs, which is $\exists \Re$-hard as shown by Kang and M{\"u}ller~\cite{kang2012sphere}. Let $G=(V,E)$ be a graph with $n$ vertices. We are going to construct $\binom{n}{2}$ graphs such that $G$ is a unit disk graph if and only if at least one of these new graphs is in $\Pi^2$. Let $S$ and $S'$ be two stars with internal vertex $s$ and $s'$ respectively, having $14n+8$ leaves each. Let $W$ and $W'$ be two paths with $2n$ vertices each with end vertices $w_1,w_{2n}$ and $w'_1,w'_{2n}$ respectively. Let $u$ and $v$ be two vertices in $V$. We define $G_{u,v}$ as the graph obtained by connecting $s$ to $s'$, $w_1$ to $u$, $w'_1$ to $v$, $w_{2n}$ to $s$ and $w'_{2n}$ to $s'$. We claim that $G$ is a unit disk graph if and only if $G_{u,v}$ is in $\Pi^2$ for some $u$ and $v$ in $V$. First let us assume that $G$ is a unit disk graph. Let us consider the set $P$ of the centres of the unit disks in any fixed representation of $G$. Consider two extreme points in $P$, meaning that removing any of them modifies the convex hull of the point set. Take the two unit disks $\D$ and $\D'$ corresponding to those two extreme points, and let us denote by $u$ and $v$ the corresponding vertices in $G$. Now take two sets $\{\D_i\}_{1 \leq i \leq 2n}$ and $\{\D'_j\}_{1 \leq j \leq 2n}$ of $2n$ unit disks such that $G(\{\D_i\})$ and $G(\{\D'_j\})$ are paths, and such that no two unit disks of the form $\D_i,\D'_j$ intersect. Moreover we require that $G(\{\D_i\})\cap G = (\{u\}, \emptyset)$ and $G(\{\D_j\})\cap G = (\{v\},\emptyset)$, and that all unit disks centres are on the same side of the line $(c_{2n},c'_{2n})$, which are the centres of $\D_{2n}$ and $\D'_{2n}$ respectively.  This is possible because the most distant points in the unit disk representation of $G$ have distance at most $4n$, and we have $2n$ unit disks in each path. Then we translate and rotate everything so that the $y$-coordinate of $c_{2n}$ and $c'_{2n}$ is equal to $2$, and that all other centres are above the horizontal line with height $2$. We take two intersecting $2$-pancakes such that one also intersect with $\D_{2n}$ and the other with $\D'_{2n}$. We choose these $2$-pancakes big enough so that for each of them we can add $14n+8$ pairwise non intersecting unit disks, but intersecting with their respective $2$-pancake. This shows that if $G$ is a unit disk graph, then $G_{u,v}$ is in $\Pi^2$.

Let us now assume that $G_{u,v}$ is in $\Pi^2$, for some $u,v$ in $V$. As a unit disk can intersect at most with $5$ pairwise non intersecting unit disks, we have that in any $\Pi^2$ representation of $G_{u,v}$, $s$ and $s'$ must be represented by $2$-pancakes, denoted by $P$ and $P'$ respectively. Let $x$ be the length of the line segment obtained as the intersection of $P$ and $Ox$. Note that all points of a unit disk intersecting a $2$-pancake are within distance $3$ of $Ox$. Therefore, the unit disks corresponding to leaves of $s$ are contained in a rectangle with area $6(x+4)$. Moreover, for each $2$-pancake intersecting $P$, there is a unit disk contained in this $2$-pancake that intersects $P$. We have $14n+8$ pairwise non-intersecting unit disks in a rectangle with area $6(x+4)$. As the area of a unit disk is bigger than $3$, we have $6(x+4)\geq 3(14n+8)$, or equivalently $x\geq 7n$. Note that the same holds with $P'$. Let us show that in any $\Pi^2$ representation of $G_{u,v}$, all the vertices in $V$ are represented by unit disks. Assume by contradiction that it is not the case. Without loss of generality, let us assume $P$ is to the left of $P'$, and that one vertex $u_G$ in $V$ is represented by a $2$-pancake that is to the right of $P'$. Indeed this $2$-pancake cannot be between $P$ and $P'$ because they are intersecting. Let us consider the last vertex in a path from $s$ to $u_G$ that is a disk. By construction, the distance between $P$ and the unit disk corresponding to this vertex is at most $2(2n+n-1)=6n-2$. This shows that this vertex is still far from the right end of $P'$, and so the next vertex has to be represented by a unit disk because it is not intersecting $P'$, which is a contradiction. We have shown that $G$ is a unit disk graph if and only if there exist $u,v$ in $V$ such that $G_{u,v}$ is in $\Pi^2$, and the construction of these $\binom{n}{2}$ graphs takes linear time for each of them.
\end{proof}

\section{Intersection graphs of convex pseudo-disks}

In this section we are interested in computing a maximum clique in intersection graphs of convex pseudo-disks. Our proof relies on line transversal and their geometric permutations on the three convex pseudo-disks that form a triangle in the complement, denoted by $\D_1$, $\D_2$ and $\D_3$. As there are only three sets, the geometric permutation of a line transversal is given simply by stating which set is the second one intersected.

\begin{definition}\normalfont
A \emph{line transversal} $\ell$ is a line that goes through the three convex pseudo-disks $\D_1$, $\D_2$ and $\D_3$. We call \emph{(convex pseudo-)disk in the middle} of a line transversal the convex pseudo-disk it intersects in second position.
\end{definition}

For sake of readability, we from now on omit to mention that a disk in the middle is a convex pseudo-disk, and simply refer to it as disk in the middle. We are going to conduct a case analysis depending on the number of convex pseudo-disks being the disk in the middle for some line transversal. When there exists no line transversal, we can prove a stronger statement. 

\begin{lemma}\label{lemma:noTransversal}
If there is no line transversal through a family of convex sets $F$, then for any pair of convex sets $\{C_1,C_2\}$ that fully intersects with $F$, $C_1$ and $C_2$ intersect.
\end{lemma}

\begin{proof}
Let us prove the contrapositive. Assume that $C_1$ and $C_2$ do not intersect, therefore there exists a separating line. As all sets in $F$ intersect $C_1$ and $C_2$, they also intersect the separating line, which is thus a line transversal of $F$.
\end{proof}

Using the notation of Theorem~\ref{thm:convexPseudoDisks}, Lemma~\ref{lemma:noTransversal} immediately implies that if there is no line transversal through the sets representing $H$, then $G\setminus H$ is a clique, which is an even stronger statement than required.
Let $\D_1$, $\D_2$ and $\D_3$ be three convex pseudo-disks that do not pairwise intersect. The following Lemma is illustrated in Figure~\ref{triangleSegmentIntersect}.
\input{3disks/triangleSegmentIntersect.tex}

\begin{lemma}\label{lemma:triangleSegmentIntersect}
Assume that $\D_1$ nor $\D_3$ is the disk in the middle of a line transversal. Let $\D$ and $\D'$ be two disks fully intersecting with $\{\D_1,\D_2,\D_3\}$. Let $p_i$ be in $\D_i\cap \D$ and let $p'_i$ be in $\D_i\cap \D'$, $1\leq i \leq 3$. We assume that the line $(p_1,p_3)$ is horizontal, and that $\D_2$ is below it. We also assume that $[p'_1,p'_3]$ is above $\D_2$. Under those assumptions, $\D$ and $\D'$ intersect.
\end{lemma}

\begin{proof}
Assume for a contradiction that $\D$ and $\D'$ do not intersect. Thus, there is a separating line $\ell$. The separating line intersects with $[p_i,p'_i]$, $1 \leq i \leq 3$, and by convexity it is a line transversal of $\{ \D_1, \D_2, \D_3\}$. By assumption, its intersection with $\D_2$ is below the line $(p_1,p_3)$. However, as $[p'_1,p'_3]$ is above $\D_2$, this implies that $\D_1$ or $\D_3$ is the disk in the middle of $\ell$, which is a contradiction.
\end{proof}

\begin{lemma}\label{lemma:triangleSegmentIntersectBis}
Assume that $\D_1$ nor $\D_3$ is the disk in the middle of a line transversal. Let $\D$ and $\D'$ be two disks fully intersecting with $\{\D_1,\D_2,\D_3\}$. Let $p_i$ be in $\D_i\cap \D$ and let $p'_i$ be in $\D_i\cap \D'$, $1\leq i \leq 3$. We assume that the line $(p_1,p_3)$ is horizontal, and that $\D_2$ is below it. We also assume that $[p'_1,p'_3]$ splits $\D_2$ into two parts, and that $\D'$ contains the part of $\D_2$ above $[p'_1,p'_3]$. Under those assumptions, $\D$ and $\D'$ intersect.
\end{lemma}

\begin{proof}
The proof is similar to the one of Lemma~\ref{lemma:triangleSegmentIntersect}. Assume for a contradiction that $\D$ and $\D'$ do not intersect. Thus, there is a separating line $\ell$. The separating line intersects with $[p_i,p'_i]$, $1 \leq i \leq 3$, and by convexity it is a line transversal of $\{ \D_1, \D_2, \D_3\}$. By assumption, its intersection with $\D_2$ is below the line $(p_1,p_3)$. Moreover, as no line transversal has $\D_1$ nor $\D_3$ as disk in the middle, it implies that $\ell$ splits $\D_2$ into two parts. By construction, $p_2$ is above $\ell$, which itself is above $[p'_1,p'_3]$. Thus $\D'$ contains $p_2$, which is a contradiction.
\end{proof}

We denote by $\{ \D'_j\}$ a set of convex pseudo-disks that fully intersect with $\{\D_1,\D_2,\D_3\}$. Our aim is to show that $G(\{ \D'_j\})$ is cobipartite.

\begin{lemma}\label{lemma:oneTransversal}
If there exists one convex pseudo-disk $\D_i \in \{\D_1,\D_2,\D_3\}$ such that the disk in the middle of all line transversals of $\{\D_1,\D_2,\D_3\}$ is $\D_i$, then $G(\{ \D'_j\})$ is cobipartite. 
\end{lemma}

\begin{proof}
Without loss of generality, let us assume that the disk in the middle of all line transversals is $\D_2$. Let $\ell$ be a line transversal, that we assume to be horizontal. Let $\D'$ be a convex pseudo-disk intersecting pairwise with $\D_1$, $\D_2$ and $\D_3$. We denote by $p'_1$ a point in $\D' \cap \D_1$ and by $p'_3$ a point in $\D' \cap \D_3$. If the line segment $[p'_1,p'_3]$ intersect $\D_2$, we have the following: Since $\D'$ and $\D_2$ are pseudo-disks, then $\D'$ must either contain the whole part of $\D_2$ that is above or the one that is below the line $(p'_1,p'_3)$. We partition the convex pseudo-disks in $\{ \D'_j\}$ into four sets depending on the line segment $[p'_1,p'_3]$.

\begin{enumerate}
    \item $[p'_1,p'_3]$ is above $\D_2$,
    \item $[p'_1,p'_3]$ intersects $\D_2$ and $\D'$ contains the whole part of $\D_2$ above it,
   \item $[p'_1,p'_3]$ is below $\D_2$,
   \item $[p'_1,p'_3]$ intersects $\D_2$ and $\D'$ contains the whole part of $\D_2$ below it.
\end{enumerate}

We are going to show that the set $X_1\subseteq \{ \D'_j\}$ of convex pseudo-disks in case $1$ or $2$ all pairwise intersect. By symmetry, the same holds for the set $X_2\subseteq \{ \D'_j\}$ of convex pseudo-disks in cases $3$ and $4$, and thus the claim will follow.

Let us suppose that we have two convex pseudo-disks $\D'$ and $\D''$ in case $1$. We can then apply Lemma~\ref{lemma:triangleSegmentIntersect} to show that $\D'$ and $\D''$ intersect. Likewise, if one is in case $1$ and the other in case $2$ we apply Lemma~\ref{lemma:triangleSegmentIntersectBis}.

Let us assume that $\D'$ and $\D''$ are in case $2$. If the line segments $[p'_1,p'_3]$ and $[p''_1,p''_3]$ intersect then it is done by convexity. Therefore we can assume without loss of generality that $[p'_1,p'_3]\cap \D_2$ is above $[p''_1,p''_3]\cap \D_2$. Hence both $\D'$ and $\D''$ contain $[p'_1,p'_3]\cap \D_2$, which shows that they intersect.

\end{proof}


\begin{definition}\normalfont
Let $\D$ and $\tilde{\D}$ be two non-intersecting disks and let $p,q$ be in the interior of $\D, \tilde{\D}$ respectively. We call \emph{external tangents} of $\D$ and $\tilde{\D}$ the two tangents that do not cross the line segment $[p,q]$.
\end{definition}

\input{3disks/diskContained.tex}

\begin{definition}\normalfont
Let $\D_i$ be a disk in $\{\D_1,\D_2,\D_3\}$, such that it is the disk in the middle of a line transversal. For sake of simplicity, let us assume that this disk is $\D_2$. We denote by $\tau$ and $\tau'$ the two external tangents of $\D_1$ and $\D_3$. We say that $\D_1$ is \emph{contained} if it is included in the surface $S$ delimited by $\D_1$, $\tau$, $\D_2$ and $\tau'$. If the intersection between $\D_2$ and exactly one of the external tangents is not empty, we say that $\D_2$ is \emph{$1$-intersecting}. If the intersection with both external tangents is not empty, we say that $\D_2$ is \emph{$2$-intersecting}. The different cases are illustrated in Figure~\ref{diskContained}. 
\end{definition}

\begin{lemma}\label{lemma:diskContained}
Let $\D_i$ be a disk in $\{\D_1,\D_2,\D_3\}$. If $\D_i$ is contained, then $G(\{ \D'_j\})$ is cobipartite. 
\end{lemma}

\input{3disks/restriction}
\input{3disks/contained}

\begin{proof}
Without loss of generality, let us assume that the considered disk is $\D_2$. By assumption, there is a line transversal having $\D_2$ as disk in the middle. Let $\ell$ be such a line transversal, and let us assume that it is horizontal. Let $\D'$ be a disk that pairwise intersect with $\D_1$, $\D_2$ and $\D_3$. We are going to define some special points. If there exists a point $q_1\in \D'\cap \D_1$ and $q_3\in \D' \cap \D_3$ such that $[q_1,q_3]$ and $\D_2$ intersect, we define $p'_1$ as $q_1$, and $p'_3$ as $q_3$. Otherwise, by convexity, for any such $q_1$ and $q_3$, the segment $[q_1,q_3]$ is always above, or always below the disk $\D_2$. Then we define $p'_1$ (respectively $p'_3$) as the intersection of the boundaries of $\D'$ and $\D_1$ (respectively $\D_3$) that is the closest to $\D_2$. As in the proof of Lemma~\ref{lemma:oneTransversal}, we are going to partition the disks in $\{ \D'_j\}$ into different cases.

Before defining the cases, we assume some non-degeneracy restriction on the disks in $\{ \D'_j\}$. Actually, the reason for this restriction is only because it allows us to define more easily the different cases. Not taking this assumption would only make the definition messier, but would not change the ideas of the proof. The following explanation is illustrated in Figure~\ref{fig:restriction}. Let $\D'$ be a disk in $\{ \D'_j\}$, such that $[p'_1,p'_3]$ is below $\D_2$ and for any $p'_2$ in $\D' \cap \D_2$, $[p'_1,p'_2] \cap \D_3 \neq \emptyset$. We assume that for any $p'_2$ in $\D' \cap \D_2$, the subset of $\D_3\setminus \D'$ that is above $[p'_1,p'_2]$ is not empty. Equivalently, we assume that if that for some disk $\D'$, it is possible to continuously move $p'_2$ in $\D' \cap \D_2$ such that $[p'_1,p'_2] \cap \D_3$ becomes a single point, then we assume that it is possible to move further $p'_2$ in $\D'\cap \D_2$ such that $[p'_1,p'_2] \cap \D_3$ is empty. This can be done by enlarging $\D'$ a little, such that it does not intersect any new disk in $\{ \D'_j\}$. We extend this assumption when the indices $1$ and $3$ are switched in the definition. We also extend this assumption for when the words ``above'' and ``below'' are switched in the definition.

We distinguish $6$ different cases.
\begin{enumerate}
    \item $[p'_1,p'_3]$ is above $\D_2$ and there exists $p'_2$ in $\D' \cap \D_2$, such that $([p'_1,p'_2] \cap \D_3 = \emptyset) \textit{ and } ([p'_3,p'_2] \cap \D_1 = \emptyset)$,
    \item $[p'_1,p'_3]$ intersects $\D_2$ and $\D'$ contains the whole part of $\D_2$ above it,
    \item $[p'_1,p'_3]$ is below $\D_2$ and for any $p'_2$ in $\D' \cap \D_2$, $([p'_1,p'_2] \cap \D_3 \neq \emptyset) \textit{ or } ([p'_3,p'_2] \cap \D_1 \neq \emptyset)$,
     \item $[p'_1,p'_3]$ is below $\D_2$ and there exists $p'_2$ in $\D' \cap \D_2$, such that $([p'_1,p'_2] \cap \D_3 = \emptyset) \textit{ and } ([p'_3,p'_2] \cap \D_1 = \emptyset)$,
    \item $[p'_1,p'_3]$ intersects $\D_2$ and $\D'$ contains the whole part of $\D_2$ below it,
    \item $[p'_1,p'_3]$ is above $\D_2$ and for any $p'_2$ in $\D' \cap \D_2$, $([p'_1,p'_2] \cap \D_3 \neq \emptyset) \textit{ or } ([p'_3,p'_2] \cap \D_1 \neq \emptyset)$.

\end{enumerate}

We are going to prove that the disks in the set $X_1$ of all disks in case $1$, $2$ or $3$, pairwise intersect. By symmetry, the same also holds for $X_2$, the set of all disks in case $4$, $5$ or $6$. Let $\D'$ and $\D''$ be two disks in $X_1$. If both correspond to case $1$ or $2$, we can apply the same reasoning as in cases $1$ and $2$ of Lemma~\ref{lemma:oneTransversal} to show that they intersect. Indeed, although for a disk in case $1$ we cannot use the fact that there is no line transversal, we have $([p'_1,p'_2] \cap \D_3 \neq \emptyset) \cup ([p'_3,p'_2] \cap \D_1 \neq \emptyset)$, thus we can still apply Lemmas~\ref{lemma:triangleSegmentIntersect} and~\ref{lemma:triangleSegmentIntersectBis}.

Let $\D'$ be a disk in case $3$, and let $\D''\in \{ \D'_j\}$ be a disk not intersecting $\D'$. We show that $\D''$ is in $X_2$, which proves the claim. The proof is illustrated in Figure~\ref{fig:contained}. We have that for any $p'_2$ in $\D' \cap \D_2$, $([p'_1,p'_2] \cap \D_3 \neq \emptyset)$ or $([p'_3,p'_2] \cap \D_1 \neq \emptyset)$. The aim of the following paragraph is to show that when shifting continuously a point $p'_2$ in $\D'\cap \D_2$, we cannot have $[p'_1,p'_2] \cap \D_3 \neq \emptyset$ to begin with, and then suddenly $[p'_3,p'_2] \cap \D_1 \neq \emptyset$. This implies that without loss of generality, we can assume that for any $p'_2$ in $\D' \cap \D_2$, we have $[p'_1,p'_2] \cap \D_3 \neq \emptyset$. 

First, we claim that for any $p'_2$ in $\D' \cap \D_2$, we have $([p'_1,p'_2] \cap \D_3 = \emptyset)$ or $([p'_3,p'_2] \cap \D_1 = \emptyset)$, and we prove it by contradiction. If such a $p'_2$ existed, there would be a point $p_3$ in $[p'_1,p'_2]\cap \D_3$ and a point $p_1$ in $[p'_3,p'_2]\cap \D_1$. Thus the segments $[p_1, p'_1]$ and $[p_3, p'_3]$ intersect at a point that is in $\D_1 \cap \D_3$, which is a contradiction. 

Furthermore, the subset of $\D_2$ for which we have $[p'_1,p'_2] \cap \D_3 \neq \emptyset$ is closed, as is the subset of $\D_2$ for which we have $[p'_2,p'_3] \cap \D_1 \neq \emptyset$. Therefore the set $\D' \cap \D_2$, which is closed, is the union of two closed sets whose intersection is empty. This implies that one of these sets is empty. Thus, we can assume without loss of generality that for any $p'_2$ in $\D' \cap \D_2$, we have $[p'_1,p'_2] \cap \D_3 \neq \emptyset$.

Let us fix one point $p'_2$ in $\D'\cap \D_2$. The segment $[p'_1,p'_2]$ splits $\D_3$ into two open parts, one of them being contained in $\D'$. As $[p'_1,p'_3]$ is below $\D_2$, the part of $\D_3$ that is not contained in in $\D'$ is also below $\D_2$. Let $p''_3$ be the point in $\D'' \cap \D_3$ as defined earlier. As we assume that $\D'$ and $\D''$ do not intersect, $p''_3$ is below $\D_2$. Let $p''_1$ be in $\D'' \cap \D_1$ as defined earlier. There are two cases, either $[p''_1,p''_3]$ is below $\D_2$, or $[p''_1,p''_3]$ intersects $\D_2$. If they intersect, as $\D''$ does not contain $p'_2$, $\D''$ contains the part of $\D_2$ that is below $[p''_1,p''_3]$. This means that $\D''$ is in case $5$, which is the claim. Finally, let us assume that $[p''_1,p''_3]$ is below $\D_2$. Assume by contradiction that $\D''$ is in case 3 (recall that it cannot be in case 1 or 2). If it is in case 3 because for any $p''_2 \in \D'' \cap \D_2$, we have $[p''_1,p''_2]\cap \D_3 \neq \emptyset$, then $\D''$ contains $p'_3$. Otherwise, for any $p''_2 \in \D''\cap \D_2$, we have $[p''_2,p''_3]\cap \D_1 \neq \emptyset$, and then $\D''$ contains $p'_1$. These two facts follow from our assumption that we have taken at the beginning of the proof, and which is illustrated in Figure~\ref{fig:restriction}. Indeed, as two pseudo-disks intersect at most twice, if $\D''$ is in case 3 then it contains either all the lower part of $\D_1$ or the lower part of $\D_3$. In any case, we have a contradiction.
\end{proof}

\begin{lemma}\label{lemma:tangentsIntersecting}
If a disk $\D_i$ is $2$-intersecting, then it is the disk in the middle of all line transversals.
\end{lemma}

\begin{proof}
Without loss of generality let us assume that this disk is $\D_2$. By definition, it is the disk in the middle of a line transversal. We denote by $\tau$ and $\tau'$ the external tangents. let $p$ be a point in $\D_2 \cap \tau$ and $p'$ be in $\D_2 \cap \tau'$. The line segment $[p,p']$ is included in $\D_2$, and separates $\D_1$ from $\D_3$. Let $\ell$ be a line transversal. Let $p_1$ be in $\ell \cap \D_1$ and $p_3$ be in $\ell \cap \D_3$. The line segment $[p_1,p_3]$ must cross $[p,p']$, which shows that the disk in the middle of $\ell$ is $\D_2$.
\end{proof}

\input{3disks/1intersecting.tex}

The following definition is illustrated in Figure~\ref{1intersecting}.

\begin{definition}\normalfont
\label{def:1intersecting}
Let $\D_i$ be a $1$-intersecting disk. We denote by $\tau_i$ the external tangents of the two other disks that $\D_i$ intersects. We denote by $A_i$ the part of $\D_i$ that is on the same side of $\tau_i$ as the two other disks.
\end{definition}

\input{3disks/outside-containing.tex}

The following definition is illustrated in Figures~\ref{outsideContainingCentred} and~\ref{outsideContainingNotCentred}.

\begin{definition}\normalfont
\label{def:outsideContaining}
Let $\D_i$ in $\{\D_1, \D_2, \D_3\}$ be a disk that is $1$-intersecting, say $\D_i=\D_2$. Let $\D'$ be a disk intersecting pairwise with $\D_1$, $\D_2$ and $\D_3$. We say that $\D'$ is \emph{outside-containing} $\D_2$ if $\D_2 \setminus A_2$ is a subset of $\D'$. We denote by $\chi_1$ and $\chi_2$ the points where the boundaries of $\D'$ and $\D_2$ intersect. Note that they are both in $A_2$. We denote by $\halfplane$ the closed halfplane with bounding line $(\chi_1,\chi_2)$ that contains $\D_2 \setminus A_2$. Let $\halfplane'$ be the closed halfplane with bounding line $\tau_2$ that contains $A_2$. Note that $(\halfplane \cap \halfplane') \setminus A_2$ is the union of one or two connected sets. We have $\D' \cap \D_1 \subset (\halfplane \cap \halfplane') \setminus A_2$ and $\D' \cap \D_3 \subset(\halfplane \cap \halfplane') \setminus A_2$. If $\D' \cap \D_1$ and $\D' \cap \D_3$ are not in the same connected set, we say that $\D'$ is \emph{centred with respect to} $\D_2$.
\end{definition}

\begin{lemma}\label{lemma:centredCentred}
Let $\D'$ and $\D''$ be intersecting with $\{\D_1, \D_2,\D_3\}$. If $\D'$ and $\D''$ are respectively centred with respect to $\D_i$ and $\D_j$, $i,j \in\{1,2,3\}$, then they intersect.
\end{lemma}

\input{3disks/centredCentred}

\begin{proof}
The proof is illustrated in Figure~\ref{centredCentred}. Let $\D'$ and $\D''$ be two disks in that are centred. If they both contain the same subset $\D_i \setminus A_i$, then they intersect. Otherwise we have $i\neq j$, let us assume without loss of generality that $\D'$ is centred with respect to $\D_1$ and $\D''$ is centred with respect to $\D_2$. There are two intersections between the boundaries of $\D''$ and $\D_2$, that we denote by $\chi''_1$ and $\chi''_2$. We denote by $\halfplane''$ the closed halfplane with bounding line $(\chi''_1,\chi''_2)$ that contains $\D_2 \setminus A_2$. We denote by $\halfplane_2$ the closed halfplane with bounding line $\tau_2$ that contains $A_2$. By assumption, $\D_1$ intersects only one of the two connected sets of $(\halfplane'' \cap \halfplane_2) \setminus A_2$. Let us consider the intersections of $\tau_2$ with the boundary of $\D_2$. By what we just said, there is a closest intersection to $\D_1$, that we denote by $p''_2$. Note that $p''_2$ is in $\halfplane''$, and therefore in $\D''$. Let $p''_1$ be a point in $\D'' \cap D_1$. If $p''_1$ is in $\D'$ then we are done. Let us now assume that it is not the case. We denote by $\chi'_1$ and $\chi'_2$ the intersections of $\D'$ with the boundary of $\D_1$. Without loss of generality, we can assume that $p''_1$ is on the boundary of $\D_1$. We denote by $p'_1$ the intersection of $\tau_1$ and the boundary of $\D_1$ that is the closest to $\D_2$, which can be defined similarly to how we defined $p''_2$. Now observe that one of $\chi'_1$ and $\chi'_2$ is between $p''_1$ and $p'_1$ on the boundary of $\D_1$. Assume without loss of generality that $\chi'_2$ is the closest to $\D_2$. Let us consider the halfplane $\halfplane'$ with bounding line $(\chi'_1,\chi'_2)$ that contains $\D_1 \setminus A_1$. We also denote by $\halfplane_1$ the halfplane with bounding line $\tau_1$ that contains $A_1$. As $\D'$ is centred with respect to $\D_1$, there is one of the two connected component that intersects with $\D_2$, and the other with $\D_3$. Note that the connected component on the side of $\chi'_2$ cannot intersect with $\D_3$, since otherwise $\D_1$ is the disk in the middle of $\tau_1$. This implies that the connected component on the side of $\chi'_2$ is the one that intersects $\D_2$. Finally, observe that either $\D'$ contains $p''_2$, or it does not intersect with $A_2$, and thus contains a point in $\D_2 \setminus A_2$. In both cases, $\D'$ contains a point in $\D''$.
\end{proof}

\begin{lemma}\label{lemma:notCentred}
Let $\D'$ and $\D''$ be intersecting with $\{\D_1, \D_2,\D_3\}$. Assume that $\D'$ is outside-containing a disk $\D_i$ in $\{\D_1, \D_2,\D_3\}$ but is not centred with respect to $\D_i$. If $\D''$ is not outside-containing any disk in $\{\D_1, \D_2,\D_3\}$, or if $\D''$ is outside-containing some disks but not centred with respect to any of them, then $\D'$ and $\D''$ intersect. 
\end{lemma}

\input{3disks/notCentred}

\begin{proof}
 The following proof is illustrated in Figure~\ref{fig:notcentred}. Let us assume by contradiction that $\D'$ and $\D''$ do not intersect. Without loss of generality, let us assume that $\D'$ is outside-containing $\D_1$. We denote by $\chi_1$ and $\chi_2$ the points where the boundaries of $\D'$ and $\D_1$ intersect. We denote by $\halfplane$ the halfplane with bounding line $(\chi_1,\chi_2)$ that contains $\D_1 \setminus A_1$. We denote by $\halfplane'$ the halfplane with bounding line $\tau_1$ that contains $A_1$. By assumption, $\D' \cap \D_2$ and $\D' \cap \D_3$ are subsets of $\halfplane \cap \halfplane'$. Because of the facts that $\D''$ is not centred and that no pair of disks in $\{\D_1, \D_2,\D_3\}$ intersect, the intersection of $\tau_1$ and one of $\D_2,\D_3$ is not in $\halfplane$. Without loss of generality, let us assume that $\tau_1 \cap \D_2$ is not in $\halfplane$. Let $q$ be in $\tau_1 \cap \D_2$, and let $p'_2$ be in $\D' \cap \D_2$. By assumption, $p'_2$ is in $\halfplane$. Let $p''_1$ be in $\D''\cap \D_1$ and $p''_3$ be in $\D'' \cap \D_3$. We have that $[p''_1,p''_3]$ and $[q,p'_2]$ intersect, which implies that $[p''_1,p''_3]$ splits $\D_2$ into two parts, one of them being contained in $\halfplane$. As we are assuming by contradiction that $\D'$ and $\D''$ do not intersect, it means that $\D''$ is outside-containing $\D_2$. Moreover, as $p''_1$ and $p''_3$ are on different sides, it implies that $\D''$ is centred with respect to $\D_2$, which is a contradiction.
\end{proof}

\begin{lemma}\label{lemma:twoTransversal}
If there exists exactly one disk, say $\D_1$, such that there is no line transversal having $\D_1$ as disk in the middle, then $G(\{ \D'_j\})$ is cobipartite. 
\end{lemma}

\input{3disks/2transversal.tex}

\begin{proof}
There is a line transversal having $\D_2$ as disk in the middle. If $\D_2$ is contained then we apply Lemma~\ref{lemma:diskContained} to conclude. Otherwise we know with Lemma~\ref{lemma:tangentsIntersecting} that $\D_2$ is $1$-intersecting. Likewise, we may assume that $\D_3$ is $1$-intersecting. We use the notation as in Definition~\ref{def:1intersecting}. Without loss of generality, let us assume that $\tau_2$ is horizontal, and that both $\D_1$ and $\D_3$ are below it. As $\D_2$ is not contained, $\D_2 \setminus A_2$ is not empty, and likewise with $\D_3$. We denote by $X_1$ the set of all disks in $\{ \D'_j\}$ that are centred with respect to $\D_2$ or $\D_3$. Recall that, by assumption, there is no disk centred with respect to $\D_1$. Let $X_2$ be the set of the remaining disks. By Lemma~\ref{lemma:centredCentred}, any two disks in $X_1$ intersect. If a disk in $\{ \D'_j\}$ is outside-containing a disk in $\{\D_2,\D_3\}$ but not centred with respect to $\D_2$, and not centred with respect to $\D_3$, then it is in $X_2$ and by Lemma~\ref{lemma:notCentred} it intersects with all disks in $X_2$.

Let $\D'$ and $\D''$ be two disks in $X_2$ that are not outside-containing a disk in $\{\D_2,\D_3\}$. We show that $\D'$ and $\D''$ intersect. The following proof is illustrated in Figure~\ref{twoTransversal}. Assume by contradiction that $\D'$ and $\D''$ do not intersect. Let $p'_1$ be a point in $\D' \cap \D_1$ and $p'_2$ in $\D' \cap \D_2$, and likewise for $p''_1$ and $p''_2$. The two line segments $[p'_1,p'_2]$ and $[p''_1,p''_2]$ do not intersect. Let us assume without loss of generality that $[p'_1,p'_2]$ is above $[p''_1,p''_2]$. As $\D''$ is not outside-containing $\D_3$, we have that $[p''_1,p''_2]$ does not split $\D_3$ in two parts, such that there exists a point in $A_3$ not contained in $\D''$. Let $p'_3$ be a point in $\D' \cap \D_3$. As there is no line transversal with $\D_1$ as disk in the middle, and from what we have just argued, we have that $[p'_1,p'_3]$ intersects with $A_2$, is above $p''_2$, and that it does not cross $\tau_2$. Therefore, since $\D'$ does not contain $\D_2 \setminus A_2$, it must contain the whole part of $\D_2$ that on the other side of the line $(p'_1,p'_3)$. Hence $\D'$ contains $p''_2$, which implies that $\D'$ and $\D''$ intersect.
\end{proof}

\begin{lemma}\label{lemma:threeTransversal}
If for each disk $\D_i$ there exists a line transversal with disk in the middle being $\D_i$, then $G(\{ \D'_j\})$ is cobipartite. 
\end{lemma}

\begin{proof}
If there is a disk contained as in Lemma~\ref{lemma:diskContained}, then it is done. Otherwise we know with Lemma~\ref{lemma:tangentsIntersecting} that each of the three disks $\D_1$, $\D_2$ and $\D_3$ is $1$-intersecting. Note that for each disk $\D_i$ we have $\D_i \setminus A_i \neq \emptyset$. We can assume that there is no disk in $\{ \D'_j\}$ that contains one of $\D_1$, $\D_2$ or $\D_3$. Indeed such a disk would intersect with all the other disks in $\{ \D'_j\}$, so we could add it arbitrarily to any of our two cliques.

We separate the disks in $\{ \D'_j\}$ into two subsets. Let $X_1$ the set of all disks $\D'$ in $\{ \D'_j\}$ that are centred with respect to some disk in $\{\D_1,\D_2,\D_3\}$. Let $X_2$ be defined as $\{ \D'_j\} \setminus X_1$. Using Lemma~\ref{lemma:centredCentred}, we immediately obtain that $G(X_1)$ is a complete graph.
Let $\D'$ and $\D''$ be two disks in $X_2$. If one of them is outside-containing, then it is not centred, and we can apply Lemma~\ref{lemma:notCentred} to show that $\D'$ and $\D''$ intersect.

From now on we assume that $\D'$ and $\D''$ are not outside-containing. Let us assume that there exists a disk $\D_i$ such that $\D'$ contains a point $p'_i \in A_i$, and $\D''$ contains a point $p''_i \in A_i$. Without loss of generality let us assume that this disk is $\D_2$. Let us also assume that $\tau_2$ is horizontal, and that $\D_1$ and $\D_3$ are below it. Let $p'_1$ (respectively $p''_1$) be a point in $\D' \cap \D_1$ (respectively $\D'' \cap \D_1$), and likewise for $p'_3$. Let us consider the triangles $p'_1p'_2p'_3$ and $p''_1p''_2p''_3$. We denote by $\chi'_1$ and $\chi'_2$ the points (potentially equal) where $p'_1p'_2p'_3$ intersects the boundary of $\D_2$. All these points are in $A_2$. We take these points such that the $x$-coordinate of $\chi'_2$ is not less that the one of $\chi'_1$. We do likewise for $\chi''_1$ and $\chi''_2$. Let us examine in which order they appear on the circular arc bounding $A_2$. Without loss of generality, let us assume that $\chi'_1$ appears first. If they appear in the following order $\chi'_1 \rightarrow \chi''_1 \rightarrow \chi'_2 \rightarrow \chi''_2$ then the triangles intersect by convexity of $A_2$. If the order is $\chi'_1 \rightarrow \chi'_2 \rightarrow \chi''_1 \rightarrow \chi''_2$ then the triangles intersect because $\D_1$ and $\D_3$ are not intersecting, and $\D'$ and $\D''$ are not outside-containing $\D_1$ or $\D_3$. Finally, in they appear in this order $\chi'_1 \rightarrow \chi''_1 \rightarrow \chi''_2 \rightarrow \chi'_2$, then the triangles intersect because $\D'$ is not outside-containing $\D_2$. Indeed two circles intersect at most twice, therefore $\D'$ would contain all the surface in $\D_2$ that is above $[\chi'_1,\chi'_2]$, which is impossible by definition of $X_2$. We have shown that $\D'$ and $\D''$ intersect.

From now on, we assume that there is no disk $\D_i$, $i \in \{1,2,3\}$, such that $\D'$ contains a point in $A_i$, and $\D''$ contains a point in $A_i$. Now let us assume that there exists a disk, say $\D_1$, such that $\D' \cap A_1=\D''\cap A_1 = \emptyset$. Let $p'_1$ be in $\D'\cap \D_1$ and $p''_1$ be in $\D''\cap \D_1$. We do likewise with $p'_2$ in $\D'\cap \D_2$ and $p''_2$ in $\D''\cap \D_2$. By assumption, $p'_1$ and $p''_1$ are in $\D_1 \setminus A_1$. If $[p'_1,p'_2]$ and $[p''_1,p''_2]$ intersect, we are done. Otherwise, without loss of generality, let us assume that $[p'_1,p'_2]$ is closer to $\D_3$ than is $[p''_1,p''_2]$. Let $p''_3$ be in $\D'' \cap \D_3$. Assume for a contradiction that $\D'$ and $\D''$ do not intersect. If $p'_2$ is in $\D_2 \setminus A_2$, then $[p''_1,p''_3]$ intersects twice with $\tau_2$, or $[p''_2,p''_3]$ intersects twice with $\tau_1$. In any case we have a contradiction, therefore $\D'$ and $\D''$ intersect. Now let us assume that $p'_2$ is in $A_2$, which implies that $p''_2$ is in $\D_2 \setminus A_2$. Assume for a contradiction that $\D'$ and $\D''$ do not intersect. Therefore, $[p''_1,p''_3]$ intersects twice with $\tau_2$, or $[p''_1,p''_3]$ intersects with $\D_2$. The first option is not possible, which implies that $[p''_1,p''_3]$ splits $\D_2$ in two parts. Now, either $\D''$ is outside-containing $\D_2$, or $\D''$ contains $p'_2$. In either way, we have a contradiction.

Finally, let us assume that for any $i\in \{1,2,3\}$, $\D'$ contains a point in $A_i$ if and only if $\D'' \cap A_i = \emptyset$. Without loss of generality, let us assume that $\D' \cap A_1 \neq \emptyset$, $\D' \cap A_2 \neq \emptyset$, $\D'' \cap A_1= \D'' \cap A_2= \emptyset$. Let $p'_1$ be in $A_1$, $p'_2$ be in $A_2$, $p''_1$ be in $\D_1 \setminus A_1$ and $p''_2$ be in $\D_2 \setminus A_2$. Assume for a contradiction that $\D'$ and $\D''$ do not intersect. First, let us assume that $[p'_1,p'_2]$ is closer to $\D_3$ than is $[p''_1,p''_2]$. Let $p''_3$ be in $\D'' \cap \D_3$. We have that $[p''_1,p''_3]$ intersects twice with $\tau_2$, or $[p''_2,p''_3]$ intersects twice with $\tau_1$. In any case we have a contradiction. Finally, let us assume that $[p''_1,p''_2]$ is closer to $\D_3$ than $[p'_1,p'_2]$. Let $p'_3$ be in $\D' \cap \D_3$. As $p'_1$ and $p'_3$ are on the same side of $\tau_2$, $[p'_1,p'_3]$ and $[p''_1,p''_2]$ intersect. We have shown that $\D'$ and $\D''$ intersect.
\end{proof}

\begin{proof}[Proof of Theorem~\ref{thm:convexPseudoDisks}]
We consider any fixed representation of $G$ with convex pseudo-disks. We denote by $\D_1$, $\D_2$ and $\D_3$ the three non-intersecting sets corresponding to $H$. If there is no line transversal of $\{ \D_1,\D_2, \D_3\}$, we conclude with Lemma~\ref{lemma:noTransversal}. If there exists exactly one disk $\D_i$ such that all line transversals have $\D_i$ as disk in the middle, we use Lemma~\ref{lemma:oneTransversal}. If there is exactly one disk $\D_i$ such that no line transversal has $\D_i$ as disk in the middle, we use Lemma~\ref{lemma:twoTransversal}. Finally, if for any disk in $\{ \D_1,\D_2, \D_3\}$ there is a line transversal having that disk as disk in the middle, we conclude with Lemma~\ref{lemma:threeTransversal}.
\end{proof}

\section{Proof of Theorem~\ref{Thm:conjImpliesEptas}} 
\label{sec:conjToEPTAS}

We first give some definitions. Vapnik and Chervonenkis have introduced the concept of VC-dimension in~\cite{vapnik2015uniform}. In this paper, we are only concerned with the VC-dimension of the neighbourhood of some geometric intersection graphs. In this context, the definition can be stated as follows:

\begin{definition}
Let $\mathcal{F}$ be a family of sets in $\Re^d$, and let $G$ be the intersection graph of $\mathcal{F}$. We say that $F\subseteq \mathcal{F}$ is \emph{shattered} if for every subset $X$ of $F$, there exists a vertex $v$ in $G$ that is adjacent to all vertices in $X$, and adjacent to no vertex in $F\setminus X$. The VC-dimension of the neighbourhood of $G$ is the maximum cardinality of a shattered subset of $\mathcal{F}$.

\end{definition}


We define the class $\mathcal{X}(d,\beta,K)$ as introduced by Bonamy {\em et al.} in~\cite{bonamy2018eptas}. Let $d$ and $K$ be in $\Positives$, and let $\beta$ be a real number such that $0<\beta \leq 1$. Then $\mathcal{X}(d,\beta,K)$ denotes the class of simple graphs $G$ such that the VC-dimension of the neighbourhood of $G$ is at most $d$, $\alpha(G)\geq \beta |V(G)|$, and $\text{iocp}(G)\leq K$. They show that there exist EPTAS (Efficient Polynomial-Time Approximation Scheme) for computing a maximum independent set in $\mathcal{X}(d,\beta,K)$. An EPTAS for a maximisation problem is an approximation algorithm that takes a parameter $\varepsilon>0$ and outputs a $(1-\varepsilon)$-approximation of an optimal solution, and running in $f(\varepsilon)n^{\mathcal{O}(1)}$ time.
More formally, we have the following:

\begin{theorem}[Bonamy {\em et al.}~\cite{bonamy2018eptas}]
\label{thm:EPTASmaxIndep}
For any constants $d,K \in \Positives$, $0<\beta \leq 1$, for every $\varepsilon > 0$, there is a randomised $(1-\varepsilon)$-approximation algorithm running in time $2^{\tilde{\mathcal{O}}(1/\varepsilon^3)}n^{\mathcal{O}(1)}$ for maximum independent set on graphs of $\mathcal{X}(d,\beta,K)$ with $n$ vertices.
\end{theorem}

Recently, Dvo{\v{r}}{\'a}k and Pek{\'a}rek have announced that it is not necessary to have bounded VC-dimension~\cite{dvovrak2020induced}. More explicitly, there is an EPTAS for the class $\mathcal{X}(+\infty,\beta,K)$. However, their running time dependence in $n$ is higher: $\tilde{\mathcal{O}}(n^5)$ with Dvo{\v{r}}{\'a}k and Pek{\'a}rek's algorithm compared to $\tilde{\mathcal{O}}(n^2)$ with the one of Bonamy {\em et al.} Also, Dvo{\v{r}}{\'a}k and Pek{\'a}rek do not compute the dependence in $\varepsilon$. For this reason, we prefer the algorithm of Bonamy {\em et al.}, despite the fact that we have to show bounded VC-dimension.



Theorem~\ref{thm:EPTASmaxIndep} states that there exists an EPTAS for computing a maximum independent set on graphs of $\mathcal{X}(d,\beta,K)$, for any $d,K \in \Positives$ and $0<\beta \leq 1$. Let $G$ be in $\Pi^3$. In order to prove Theorem~\ref{Thm:conjImpliesEptas}, we show that the VC-dimension of the neighbourhood of any vertex in $G$ is bounded. Observe that the VC-dimension of a graph and its complement are equal. We aim at using the EPTAS mentioned above for computing a maximum independent set in the complement, which is equivalent to computing a maximum clique in the original graph. However a graph $G$ in $\Pi^3$ does not necessarily satisfy $\alpha(\overline{G}) \geq \beta |V(G)|$ for some $0< \beta \leq 1$. Even if it does, we need to know the value of $\beta$ in order to use the EPTAS of Theorem~\ref{thm:EPTASmaxIndep}. Therefore we show how to compute a maximum clique in any $G\in \Pi^3$ by using polynomially many times the EPTAS of Theorem~\ref{thm:EPTASmaxIndep} on some subgraphs of $G$, which have the desired property.

In general, for intersection graphs of geometric objects that can be described with finitely many parameters, the VC-dimension of the neighbourhood is bounded. For graphs in $\Pi^3$, we were able to show an upper bound of $28$. We do not expect this value to be tight, but showing any constant was sufficient for our purpose.

\begin{proposition}\label{prop:VCdim28}
The VC-dimension of the neighbourhood of a graph $G=(V,E)$ in $\Pi^3$ is at most $28$.
\end{proposition}


We use the fact that the VC-dimension of the neighbourhood of disk graphs (and even pseudo-disk graphs) is at most $4$, as proved by Aronov {\em et al.}~\cite{aronov2018pseudo}. Likewise, the VC-dimension of the neighbourhood of unit ball graphs is at most $4$, as noticed by Bonamy {\em et al.}~\cite{bonamy2018eptas}. For any point $c\in \mathbb{R}^3$ and any non-negative real number $\rho$, we denote by $\B(c,\rho)$ the ball centred at $c$ with radius $\rho$. Moreover, we denote by $P^3(c,\rho)$ the $3$-pancake that is the Minkowski sum of the unit ball centred at the origin and the disk lying on the plane $xOy$, centred at $c$ with radius $\rho$. Note that if $\rho=0$, then $P^3(c,\rho)$ is the unit ball centred at $c$. Before showing Proposition~\ref{prop:VCdim28}, we show the following:

\begin{lemma}
\label{lemma:ballPancakeIntersect}
Let $\B$ be a unit ball centred at $c$ and let $P^3(c',\rho)$ be a $3$-pancake. We denote by $\D$ the disk that is the intersection of $\B(c,2)$ and the plane $xOy$. Also, we denote by $\D'$ the disk $\D(c',\rho)$ (which is a strict subset of the intersection of $P^3$ and the plane $xOy$). We have that $\B$ and $P^3$ intersect if and only if $\D$ and $\D'$ intersect. 
\end{lemma}

\begin{proof}
By definition, $\B$ and $P^3$ intersect if and only if there exists a unit ball $\B'$ whose centre lies in $\D'$ such that $\B$ and $\B'$ intersect. This is equivalent to say that $\B(c,2)$ contains a point in $\D'$. Finally, this statement is equivalent to having $\D$ and $\D'$ intersecting.
\end{proof}

\begin{proof}[Proof of Proposition~\ref{prop:VCdim28}]

First let us show that if $V$ is shattered, then in any $\Pi^3$ representation of $G$ there are at most four $3$-pancakes. Let us assume by contradiction that there exists a set $S$ of five $3$-pancakes, such that for every subset $T$ of $S$, there exists a unit ball or a $3$-pancake intersecting all elements in $T$ and intersecting no element in $S\setminus T$. For each $3$-pancake $P^3(c_i,\rho_i)$ in $S$, we denote by $\D_i$ the disk $\D(c_i,\rho_i)$ lying on the plane $xOy$. Let $T$ be a subset of $S$. If there exists a $3$-pancake $P^3(c',\rho')$ intersecting with the elements of $T$ and with no element in $S \setminus T$, we denote by $\D_T$ the disk $\D(c',\rho'+2)$ lying on the plane $xOy$. Otherwise there exists a unit ball $\B$ centred at $c''$ intersecting intersecting with the elements of $T$ and with no element in $S \setminus T$, and then we denote by $\D_T$ the intersection between $\B(c'',2)$ and $xOy$. As $\B$ intersects with a $3$-pancake, $\D_T$ is not empty. Using Lemma~\ref{lemma:ballPancakeIntersect}, we have that $\D_i$ intersects with $\D_T$ if and only if $P^3(c_i,\rho_i)$ is in $T$. This implies that if $S$ is shattered by some $3$-pancakes and unit balls, then the set $\{\D_i\}$ is shattered by $\{\D_T \mid T \subseteq S\}$. However this is not possible because the VC-dimension of the neighbourhood of disk graphs is at most $4$.

Now let us prove the claim. Assume by contradiction that we have a shattered set with $29$ elements. As shown above, in any $\Pi^3$ representation there are at least $25$ unit balls. Let us consider such a representation. We denote by $S_1, \dots, S_5$ five sets of five unit balls each. As the VC-dimension of the neighbourhood of unit ball graphs is at most $4$, for each set $S_i$ there exists a non-empty subset $T_i \subseteq S_i$ such that no unit ball can intersect with the unit balls in $T_i$, but not with those in $S_i \setminus T_i$. Therefore the absolute height of the centre of any unit ball in $T_i$ is at most $2$, since $T_i$ is realised by a $3$-pancake. For each $T_i$, we choose arbitrarily one unit ball $\B_i$, and define a new set $T$ as $\{\B_1,\dots,\B_5\}$. Moreover for each unit ball $\B_i$ centred at $c_i$, we denote by $\D_i$ the intersection between $\B(c_i,2)$ and the plane $xOy$. Note that $\D_i$ is not empty. Let $T'$ be a subset of $T$, and let us consider the set $\cup_{\B_i \in T'} T_i$, that we denote by $T'_+$. Note that unless $T'=\emptyset$, no unit ball can intersect with all elements in $T'_+$ and with no element in $S \setminus T'_+$. Therefore this can only be achieved by a $3$-pancake $P^3(c,\rho)$, and we denote by $\D_{T'}$ the disk $\D(c,\rho)$ lying on the plane $xOy$. Using Lemma~\ref{lemma:ballPancakeIntersect}, the five disks $\D_i$ are shattered by the disks in $\{D_{T'}\mid T' \subseteq T\}$, which is impossible.
\end{proof}

\begin{proof}[Proof of Theorem~\ref{Thm:conjImpliesEptas}]

Let $G$ be a graph in $\Pi^3$ with $n$ vertices. Since the VC-dimension of a graph is the same as its complement, Proposition~\ref{prop:VCdim28} implies that the VC-dimension of $\overline{G}$ is at most $28$. First let us assume that a representation of $G$ is given. For every vertex represented by a unit ball, we are going to compute a maximum clique containing this vertex. As noticed by Bonamy {\em et al.}, for any vertex $v$ represented by a unit ball, we have $|\mathcal{N}(v)|\leq 25 \omega(G)$~\cite{bonamy2018eptas}. Let us denote by $G_v$ the subgraph induced by $\mathcal{N}(v)$. Thus we have $\alpha(\overline{G}_v)\geq |\mathcal{N}(v)|/25$. This shows that $\overline{G}_v$ is in $\mathcal{X}(28,1/25,K)$. Using Theorem~\ref{thm:EPTASmaxIndep}, we have a randomised EPTAS for computing a maximum independent set in $\overline{G}_v$, which is equivalent to computing a maximum clique in $G_v$. Note that computing a maximum clique in $G_v$ for each vertex $v$ represented by a unit ball adds at most a multiplicative factor $n$ in the running time. It remains to compute a maximum clique that only contains vertices represented by $3$-pancakes. Instead of considering $3$-pancakes, one can only look at the corresponding disks on the plane $xOy$. This can be done as suggested in~\cite{bonamy2018eptas}: find four piercing points in time $\mathcal{O}(n^8)$, then consider the subgraph $H$ of disks that are pierced by at least one of these points. We have $\alpha(\overline{H})\geq n'/4$ where $n'$ denotes the number of vertices in $H$. This implies that $H$ is in $\mathcal{X}(28,1/4,K)$, and we can conclude as before. 

Now assume that a representation is not given. As we do not know whether a vertex can be represented by a unit ball, we cannot compute a maximum clique as was done above. If there exists a representation of $G$ with at least one vertex $v$ represented as a unit ball, then $\alpha(G_v)\leq 12$, because the kissing number for unit spheres is $12$. Indeed for any $3$-pancake $P^3$ intersecting a unit ball $B$, there exists a unit ball $B' \subseteq P^3$ such that $B$ and $B'$ intersect. Thus, if instead of each pancake there were such a unit ball, we would have the desired inequality. But since such a unit ball $B'$ is contained in the corresponding $3$-pancake $P^3$, the independence number of $G_v$ can only decrease when considering the actual $3$-pancakes, which implies $\alpha(G_v)\leq 12$. If there exists a representation only with $3$-pancakes, then the vertex $v$ corresponding the $3$-pancake with the smallest radius satisfies $\alpha(G_v)\leq 6$. Therefore in any case there must be a vertex $v$ with $\alpha(G_v)\leq 12$. We can find such a vertex in $\mathcal{O}(n^{13})$ time by testing for each $v$ whether there is an independent of size $12$ in $G_v$. 

In order to give a linear lower bound on $\alpha(\overline{G}_v)$, we first give an upper bound on the chromatic number of any graph in $\Pi^3$. Let $\tilde{G}$ be a graph in $\Pi^3$, given with a fixed representation. We denote by $V_1$ the set of vertices represented by unit balls, and by $V_2$ those represented by $3$-pancakes. We denote by $\tilde{G}_1$ the graph induced by $V_1$. As noted in~\cite{bonamy2018eptas}, we have for each $v_1 \in V_1$, $|\mathcal{N}(v_1)|\leq 25 \omega(\tilde{G}_1)$. Since $\omega(\tilde{G}_1)\leq \omega(\tilde{G})$, the maximum degree in $\tilde{G}_1$ is at most $25 \omega(\tilde{G})-1$, which implies that we can colour the vertices in $V_1$ using at most $25 \omega(\tilde{G})$ colours. For disk graphs, the chromatic number is at most $6$ times the clique number. Thus we can colour the vertices in $V_2$ using at most $6 \omega(\tilde{G})$ other colours. So in total we have $\chi(\tilde{G})\leq 31 \omega(\tilde{G})$.

Going back to the subgraph $G_v$, we have $\alpha(G_v) \omega(G_v) \geq \alpha(G_v) \chi(G_v)/31 \geq |\mathcal{N}(v)|/31$. Therefore we obtain $\omega(G_v)\geq  |\mathcal{N}(v)|/372$. This implies that $\overline{G}_v$ is in $\mathcal{X}(28,1/372,K)$, and therefore we have an EPTAS for computing a maximum clique containing $v$. We can iterate this process in the graph $G$ where $v$ has been removed to compute a maximum clique that does not contain $v$. As we repeat this process linearly many times, we obtain an EPTAS for computing a maximum clique in $G$.
\end{proof}

\newpage
\bibliography{bib}

\end{document}

%% file: dimension2/2pancake.tex
\begin{figure}
   \centering
    \begin{tikzpicture}[scale=1.4]
    
\draw[->] (-3,0)--(4,0);

\draw[fill,red, opacity=0.3 ] (2.82,1)--(-1.68,1) arc (90:270:1cm)--(2.82,-1) arc (-90:90:1cm);
\draw[dashed] (2.82,0) circle (1cm);
\draw[dashed] (-1.68,0) circle (1cm);

\draw[thick] (-1.68,0)--(2.82,0);

\draw[] (-3,0)  node[left] {$x$-axis};
\draw[] (0.57,0)  node[below] {$s$};

  \foreach \Point in {((2.82,0),(-1.68,0)}{
   \node at \Point {\textbullet};
}

\end{tikzpicture}
   \caption{The union of the unit disks centred at points of $s$ is a $2$-pancake.}
    \label{fig:2pancake}
\end{figure}

%% file: dimension2/1disk1pancake.tikz
    \begin{tikzpicture}[scale=1.5]
    
\draw[->] (-3,0)--(1.5,0);

\draw[red] (0,0)--(-1.44,1.38); 
\draw[red, dashed] (-1.7,1.05)--(-1.09,1.05);

\draw[blue] (-1.09,1.05) circle (1cm);
\draw[blue] (0,1) arc (90:270:1cm);
\draw[blue] (0,1)--(1.5,1);
\draw[blue] (0,-1)--(1.5,-1);

\draw[red] (0.08,2) arc (87.67:184.85:2cm);
\draw[red] (-1.99,-0.17) arc (-126.55:39.07:1.51cm);

\draw[] (0,0) node[below right] {$(x_1,0)$};
\draw[] (-2,0) node[above left] {($x_1-2,0)$};
\draw[] (-1.09,1.05) node[right] {$c$};
\draw[] (-1.7,1.05) node[left] {$c'$};
\draw[] (-1.99,-0.17) node[below left] {$r_2$};

\draw[red] (0,1.5) node[] {$R_1$};
\draw[red] (-1.09,-0.25) node[] {$R_2$};
\draw[blue] (-1.66,1.59) node[] {$\D$};
\draw[blue] (1,0.5) node[] {$P^2$};

  \foreach \Point in {(0,0),(-2,0), (-1.09,1.05),(-1.7,1.05),(-1.99,-0.17)}{
   \node at \Point {\textbullet};
}

\end{tikzpicture}

%% file: dimension2/help2disks.tex
\begin{figure}
   \centering
    \begin{tikzpicture}[scale=1.4]
    
\draw[->] (-3,0)--(1.5,0);
\draw (-3,1)--(1.5,1);

\draw[dashed] (0.32,0)--(0.32,1);

\draw[blue] (-0.59,1.41) circle (1cm);
\draw[blue] (-0.28,-0.36) circle (1cm);
\draw[red] (0.32,0) circle (1cm);
\draw[red] (-1.68,0) circle (1cm);

\draw[] (0.32,0) node[below] {$(x_\ell,0)$};
\draw[] (-1.68,0) node[below] {$(x_\ell-2,0)$};
\draw[blue] (-0.59,1.41) node[] {$\D$};
\draw[blue] (-0.28,-0.36)  node[] {$\D'$};
\draw[] (-3,0)  node[left] {$Ox$};
\draw[] (-3,1)  node[left] {$\ell$};

  \foreach \Point in {((0.32,0),(-1.68,0),(0.32,1)}{
   \node at \Point {\textbullet};
}

\end{tikzpicture}
   \caption{Illustration of the proof of Lemma~\ref{lemma:help2disks}}
    \label{help2disks}
\end{figure}

%% file: dimension2/nonEmptyIntersection.tikz
\begin{tikzpicture}[scale=1.3]

\draw[->] (-3,0)--(1,0);

\draw[blue] (-1.94,-2.01)--(-1.44,-1.66);
\draw[blue] (-0.32,-0.85)--(0.18,-0.49);

\draw[red] (-1.44,-1.66)--(-0.32,-0.85);

\draw[blue] (-1.97,0.27) arc (145.97:285.34:2cm);
\draw[blue] (0.21,-2.77) arc (-34.03:105.34:2cm);
\draw[red] (-1.58,-0.27) arc (155.66:275.66:1.39cm);
\draw[red] (-0.18,-2.23) arc (-24.34:95.66:1.39cm);

\draw[] (-2.8,0) node[above] {$Ox$};
\draw[] (-1.44,-1.66) node[right] {$c$};
\draw[] (-0.32,-0.85) node[right] {$c'$};
\draw[] (-1.58,-0.27) node[left] {$r_1$};
\draw[red] (-1.11,-0.93) node[] {$R_1$};
\draw[red] (-0.65,-1.58) node[] {$R_2$};
\draw[blue] (-1.9,-1.49) node[] {$R^+_1$};
\draw[blue] (-1.43,-2.14) node[] {$R^+_2$};

  \foreach \Point in {(-1.44,-1.66),(-0.32,-0.85),(-1.58,-0.27)}{
   \node at \Point {\textbullet};
}

\end{tikzpicture}

%% file: dimension2/emptyIntersection.tikz
\begin{tikzpicture}[scale=3.3]
    
\draw[] (-0.3,0)--(0.9,0);
\draw[olive] (0.9,0)--(1.28,0);
\draw[->] (1.28,0)--(2.5,0);
\draw (0.9,0)--(0.59,0.95);
\draw[red] (0.59,0.95)--(1.52,0.97);
\draw (1.28,0)--(1.52,0.97);

\draw[blue] (-0.4,1.06) arc (-186.37:6.37:1cm);
\draw[blue] (0.52,1.05) arc (-184.39:4.39:1cm);
\draw[] (1.89,-0.13) arc (-7.45:187.45:1cm);

\draw[red] (0.59,0.95) arc (181.15:241.15:0.93cm);
\draw[red] (1.07,0.16) arc (-58.85:1.15:0.93cm);

\draw[] (0.59,0.95) node[above right] {$c$};
\draw[] (1.52,0.97) node[above left] {$c'$};
\draw[blue] (-0.2,0.95) node[above] {$\D(c,2)$};
\draw[blue] (2.35,1.02) node[] {$\D(c',2)$};

\draw[] (-0.2,0) node[above] {$Ox$};
\draw[] (0,0.57) node[] {$\C_1$};

\draw[] (0.9,0) node[below] {$p_1$};
\draw[] (0.76,0.44) node[above right] {$p'_1$};
\draw[] (1.28,0) node[below] {$p_2$};

\draw[red] (1.07,0.42) node[] {$R_1$};
\draw[olive] (1.09,0) node[below] {$s$};

  \foreach \Point in {(0.59,0.95),(1.52,0.97),(0.9,0),(0.76,0.44),(1.28,0)}{
   \node at \Point {\textbullet};
}

\end{tikzpicture}

%% file: motivation/2disksCounterExample.tex
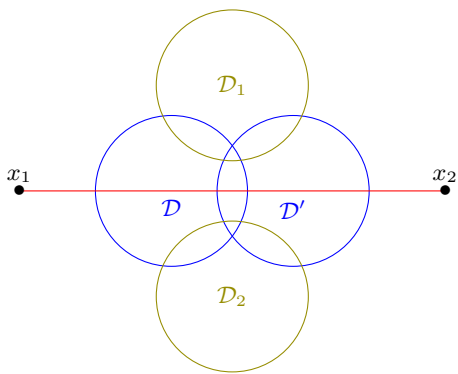
\begin{figure}
    \centering
    \begin{tikzpicture}[scale=1]

\draw[blue] (0,0) circle (1cm);
\draw[blue] (1.6,0) circle (1cm);

\draw[olive] (0.8,1.4) circle (1cm);
\draw[olive] (0.8,-1.4) circle (1cm);

\draw[red] (-2,0)--(3.6,0);

\draw[blue] (0,0) node[below] {$\D$};
\draw[blue] (1.6,0) node[below] {$\D'$};
\draw[] (-2,0) node[above] {$x_1$};
\draw[] (3.6,0) node[above] {$x_2$};

\draw[olive] (0.8,1.4) node[] {$\D_1$};
\draw[olive] (0.8,-1.4) node[] {$\D_2$};

  \foreach \Point in {(-2,0),(3.6,0)}{
   \node at \Point {\textbullet};
}

\end{tikzpicture}
    \caption{Lemma~\ref{lemma:2disks} does not hold in $\Tilde{\Pi}^2$: $G(\{\D_1,\D_2,[x_1,x_2]\})$ is not cobipartite}
    \label{2disksCounterExample}
\end{figure}

%% file: 3disks/triangleSegmentIntersect.tex
\begin{figure}
    \centering
    \begin{tikzpicture}[scale=0.9]

\draw[ ] (-1,-0.2) circle (1cm);
\draw[ ]  (2,0.2) circle (1.3cm);

\draw[ ] (0.3,-2.2) circle (0.8cm);

\draw[dashed, red] (-5,-0.8)--(-0.7,-0.8);
\draw[dashed, red] (1.7,-0.8)--(5,-0.8);

\draw[fill,blue, opacity =0.2] (-1.4,-0.5)--(1.2,0.8)--(-0.1,-2)--(-1.4,-0.5);
\draw[fill,red, opacity =0.2] (-0.7,-0.8)--(1.7,-0.8)--(0.7,-2.3)--cycle;

\draw (-1.4,-0.5) node[left] {$p'_1$};
\draw (1.2,0.8) node[right] {$p'_3$};
\draw (-0.1,-2) node[below] {$p'_2$};
\draw (0.7,-2.3) node[below] {$p_2$};
\draw (1.7,-0.8) node[above] {$p_3$};
\draw (-0.7,-0.8) node[above] {$p_1$};

   \foreach \Point in {(-0.7,-0.8),(1.7,-0.8), (-1.4,-0.5),(1.2,0.8),(-0.1,-2),(0.7,-2.3)}{
   \node at \Point {\textbullet};
}
\end{tikzpicture}
  
    \caption{Illustration of Lemma~\ref{lemma:triangleSegmentIntersect}. The triangles $p_1p_2p_3$ and $p'_1p'_2p'_3$ intersect.}
    \label{triangleSegmentIntersect}
\end{figure}
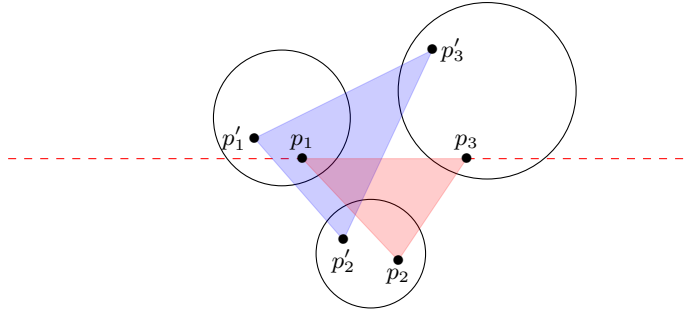

%% file: 3disks/diskContained.tex
\begin{figure}
    \centering
    \begin{tikzpicture}[scale=1]

\draw[ ] (0,0) circle (1cm);
\draw[ ]  (6.58,0.19) circle (1.7cm);
\draw[ red] (2.9,0.3) circle (1.75cm);
\draw[ blue] (4,-0.6) circle (0.6cm);
\draw[olive] (3.5,2.5) circle (1.6cm);

\draw (-1.84,0.76)--(8.88,2.22);
\draw (-2,-0.85)--(9.34,-1.73);

\draw (0,0) node[] {$\D_1$};
\draw[red] (2.9,0.3) node[] {$\D''_2$};
\draw (6.58,0.19) node[] {$\D_3$};
\draw[blue] (4,-0.6) node[] {$\D_2$};
\draw[olive] (3.5,2.5) node[] {$\D'_2$};

\draw (-1.84,0.76) node[above right] {$\tau$};
\draw (-2,-0.85) node[below right] {$\tau'$};

\draw (1.2,-0.85) node[] {$S$};

\end{tikzpicture}
    \caption{$\D_2$ is contained, $\D'_2$ is $1$-intersecting and $\D''_2$ is $2$-intersecting.}
    \label{diskContained}
\end{figure}
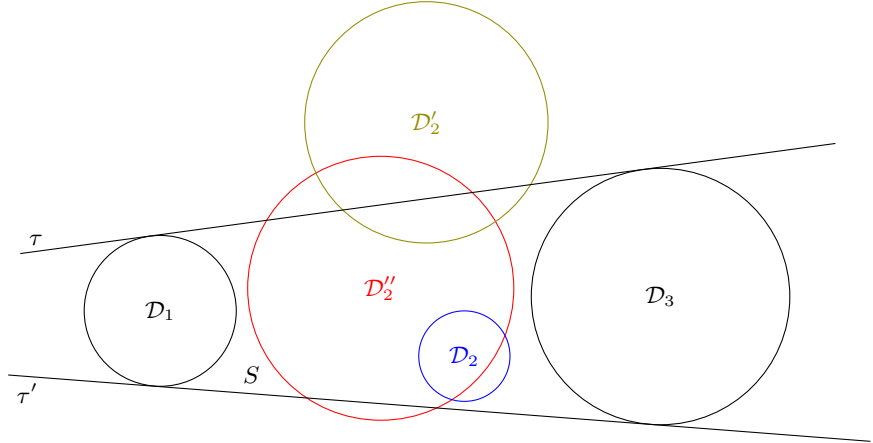

%% file: 3disks/restriction.tex
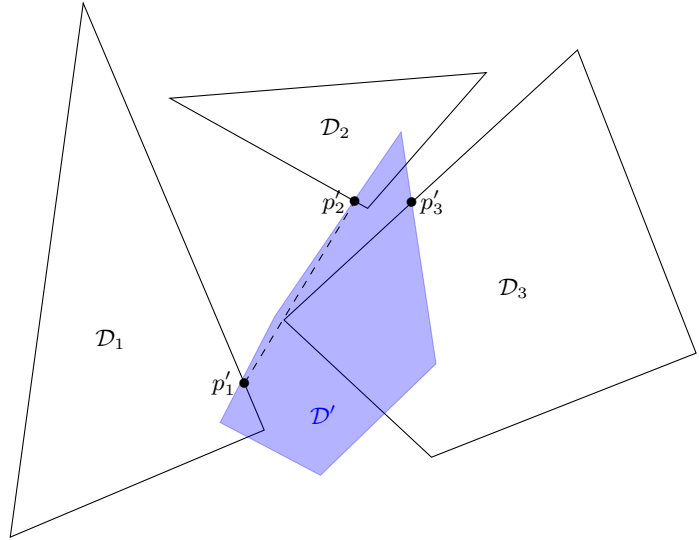
\begin{figure}
    \centering
    \begin{tikzpicture}[scale=1]
    
\draw[] (5.34,5.83) -- (1.48,2.25) -- (3.42,0.43) -- (6.9,1.81) -- cycle;
\draw[] (1.22,0.79) -- (-1.16,6.45) -- (-2.12,-0.63) -- cycle;
\draw[] (4.14,5.53) -- (-0.02,5.19) -- (2.58,3.73) -- cycle;
\draw[fill, blue, opacity=0.3] (0.64,0.89) --(1.36,2.3)-- (3.0208766437069503,4.744752661239825) -- (3.48,1.67) -- (1.96,0.19) -- cycle;

\draw[dashed] (0.96,1.41)--(2.41,3.82);

\draw[] (-0.8,1.99) node[] {$\D_1$};
\draw[] (2.16,4.79) node[] {$\D_2$};
\draw[] (4.5,2.67) node[] {$\D_3$};
\draw[blue] (2,1) node[] {$\D'$};

\draw[left] (0.96,1.41) node[] {$p'_1$};
\draw[left] (2.41,3.82) node[] {$p'_2$};
\draw[right] (3.16,3.81) node[] {$p'_3$};

  \foreach \Point in {(0.96,1.41),(2.41,3.82),(3.16,3.81)}{
   \node at \Point {\textbullet};
}

\end{tikzpicture}
    \caption{We can as well assume the existence of $p'_2$ such that $[p'_1,p'_2] \cap \D_3 = \emptyset$}
    \label{fig:restriction}
\end{figure}

%% file: 3disks/contained.tex
\begin{figure}
    \centering
    \begin{tikzpicture}[scale=1]
    
\draw[] (5.34,5.83) -- (1.48,2.25) -- (3.42,0.43) -- (6.9,1.81) -- cycle;
\draw[] (1.22,0.79) -- (-1.16,6.45) -- (-2.12,-0.63) -- cycle;
\draw[] (4.14,5.53) -- (-0.02,5.19) -- (2.58,3.73) -- cycle;
\draw[fill, blue, opacity=0.3]  (1,1) -- (3.6129847377732767,5.191967073787209) -- (6.285948302361679,6.367252455192616) -- (7.299519024079584,0.820993465952257) -- (2.499248086023586,-0.07905733493323994) -- cycle;

\draw [red] (0.68,2.08)-- (1.8,2.55);

\draw[] (-0.8,1.99) node[] {$\D_1$};
\draw[] (2.16,4.79) node[] {$\D_2$};
\draw[] (4.5,2.67) node[] {$\D_3$};
\draw[blue] (2,1) node[] {$\D'$};

\draw[left] (1.08,1.13) node[] {$p'_1$};
\draw[left] (3.53,4.92) node[] {$p'_2$};
\draw[ right] (2.19,2.91) node[] {$p'_3$};

\draw[left] (0.68,2.08) node[] {$p''_1$};
\draw[above left] (1.8,2.55)node[] {$p''_3$};

  \foreach \Point in {(1.08,1.13),(3.53,4.92),(2.19,2.91),(0.68,2.08), (1.8,2.55)}{
   \node at \Point {\textbullet};
}

\end{tikzpicture}
    \caption{Illustration of the proof of Lemma~\ref{lemma:diskContained}: The disk $\D'$ is in case 3, and so the disk $\D''$ not intersecting with $\D'$ is in case 4 or 5.}
    \label{fig:contained}
\end{figure}
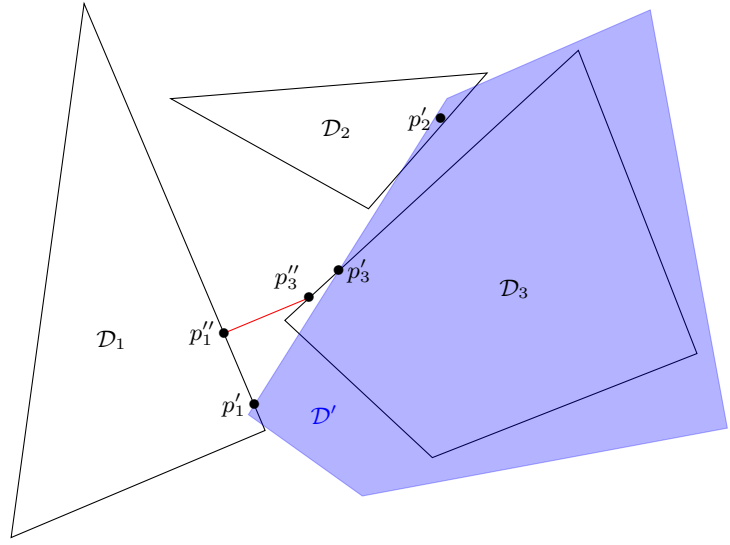

%% file: 3disks/1intersecting.tex
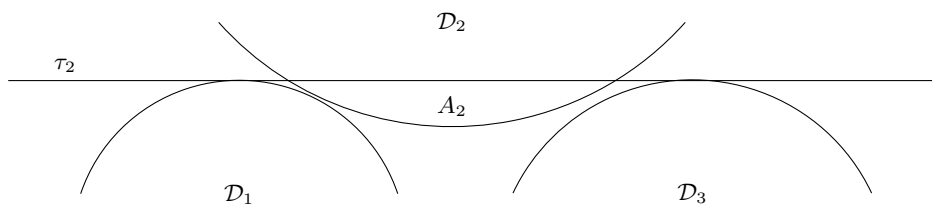
\begin{figure}
    \centering
    \begin{tikzpicture}[scale=0.75]

\draw[ ] (6.14,1.47) arc (25.01:154.99:3.47cm);
\draw[ ] (-2.17,1.46) arc (18.41:161.59:2.93cm);
\draw (-5.31,4.49) arc (221.47:318.53:5.46cm) ;

\draw (-9,3.46)--(7.2,3.46);

\draw[] (-8,3.46) node[above] {$\tau_2$};

\draw[] (-4.95,1.45) node[] {$\D_1$};
\draw[] (-1.22,4.5) node[] {$\D_2$};
\draw[] (3,1.47) node[] {$\D_3$};

\draw[] (-1.22,3) node[] {$A_2$};

\end{tikzpicture}
    \caption{Illustration of Definition~\ref{def:1intersecting}, with $\D_2$ being the $1$-intersecting disk}
    \label{1intersecting}
\end{figure}

%% file: 3disks/outside-containing.tex
\begin{figure}
    \centering
    \begin{tikzpicture}[scale=0.85]
    
\draw[ ] (-1.92,-1.47) circle (2.19cm);
\draw[ ] (2.81,-1.16) circle (2.49cm);
\draw[ ] (0.9,2.01) circle (1.18cm);

\draw[] (-1.92,-1.47) node[] {$\D_1$};
\draw[] (0.9,2.01)node[] {$\D_2$};
\draw[] (2.81,-1.16) node[] {$\D_3$};

\draw (-5.55,0.26)--(6.07,1.77);

\draw[blue] (-5.52,0.79) arc (263.04:291.34:24.3cm);
\draw[dashed, blue] (-4.78,0)--(5.2,1.52);

\draw[blue] (-5,1.4) node[] {$\D'$};

\draw[blue,below right ] (1.42,0.94) node[] {$\chi_2$};
\draw[blue] (0.43,0.44) node[] {$\chi_1$};

  \foreach \Point in {(1.42,0.94),(0.73,0.84)}{
   \node at \Point {\textbullet};
}

\end{tikzpicture}
    \caption{Illustration of Definition~\ref{def:outsideContaining}, $\D'$ is centred with respect to $\D_2$.}
    \label{outsideContainingCentred}
\end{figure}
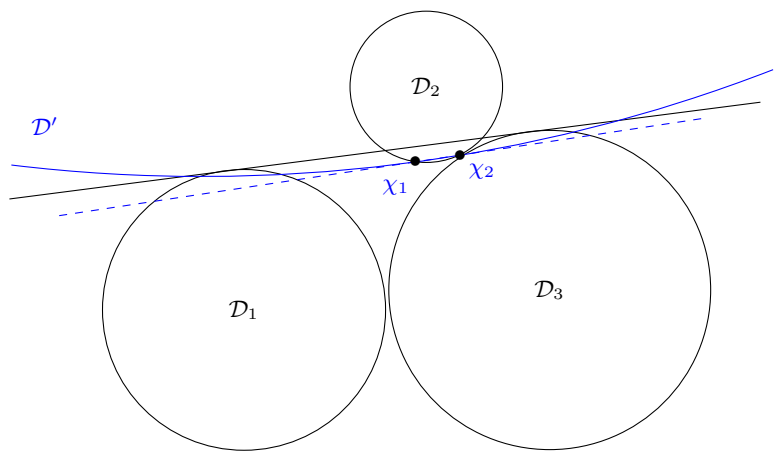

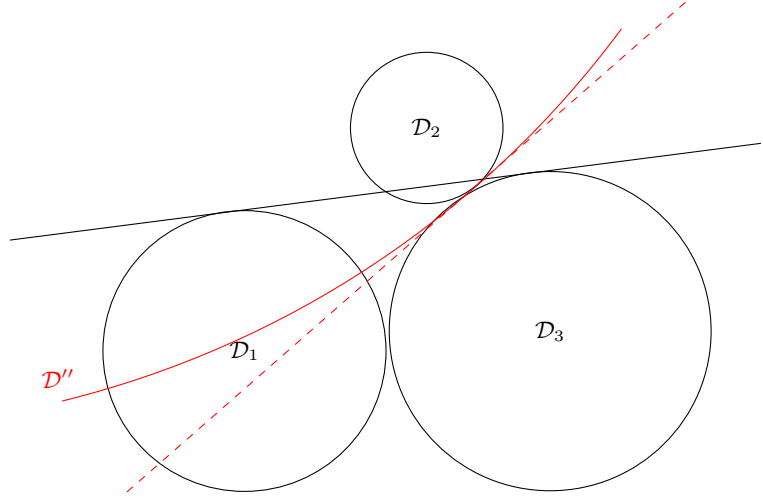
\begin{figure}
    \centering
    \begin{tikzpicture}[scale=0.85]
    
\draw[ ] (-1.92,-1.47) circle (2.19cm);
\draw[ ] (2.81,-1.16) circle (2.49cm);
\draw[ ] (0.9,2.01) circle (1.18cm);

\draw[] (-1.92,-1.47) node[] {$\D_1$};
\draw[] (0.9,2.01)node[] {$\D_2$};
\draw[] (2.81,-1.16) node[] {$\D_3$};

\draw (-5.55,0.26)--(6.07,1.77);

\draw[red] (-4.8,-1.9) node[] {$\D''$};

\draw[red] (-4.74,-2.25) arc (283.95:323.77:15.3cm);
\draw[dashed, red] (-3.74,-3.67)--(4.97,4.03);

\end{tikzpicture}
    \caption{Illustration of Definition~\ref{def:outsideContaining}, $\D''$ is outside-containing $\D_2$, but not centred with respect to $\D_2$.}
    \label{outsideContainingNotCentred}
\end{figure}

%% file: 3disks/centredCentred.tex
\begin{figure}
    \centering
    \begin{tikzpicture}[scale=1.6]

\draw  (2.67,2.8) circle (2cm);
\draw  (0.14826367619253633,1.3615091609992063)-- (6.3250362050575735,1.3968453882810428);
\draw[red]  (0.64,1.18)-- (5.7,0.32);
\draw (5.480891945225715,0.17795316749411083) circle (1.2163460030150042cm);
\draw[red]  (4.082440800432348,1.3840158951195713)-- (4.355596500067938,0.639697425885851);
\draw (4.669974092696311,2.810179810585946)-- (4.734905977374927,-2.56787931274102);
\draw[blue] (4.522888613683907,2.591209870407119)-- (4.2048625681473775,-2.7728294309756722);

\draw[] (4.9,-1.8)  node[] {$\tau_1$};
\draw[] (0.6,1.5)  node[] {$\tau_2$};
\draw[above left,red] (4.08,1.38)  node[] {$p''_2$};
\draw[below,red] (1.86,0.97)  node[] {$\chi''_1$};
\draw[below,red] (2.83,0.81)  node[] {$\chi''_2$};
\draw[left,blue] (4.35,-0.28)  node[] {$\chi'_1$};
\draw[right,blue] (4.41,0.76) node[] {$\chi'_2$};
\draw[above left,red] (4.36,0.64) node[] {$p''_1$};
\draw[right,blue] (4.69,1.1)node[] {$p'_1$};

\draw[] (2.67,2.8)  node[] {$\D_2$};
\draw[] (5.480891945225715,0.17795316749411083)  node[] {$\D_1$};

   \foreach \Point in {(4.08,1.38),(1.86,0.97),(2.83,0.81),(4.35,-0.28),(4.36,0.64),(4.41,0.76),(4.69,1.1)}{
   \node at \Point {\textbullet};
}

\end{tikzpicture}
    \caption{Illustration of the proof of Lemma~\ref{lemma:centredCentred}.}
    \label{centredCentred}
\end{figure}

%% file: 3disks/notCentred.tex
\begin{figure}
    \centering
    \begin{tikzpicture}[scale=0.6]

\draw[] (-7.0656594795748315,1.7734174219739816) -- (5.26,-2.53) -- (0.7,-4.55) -- cycle;
\draw[] (7.069239721509552,-1.6543713029142542) -- (3.02,-5.61) -- (5.38,-9.17) -- (9.26,-5.95) -- cycle;
\draw [] (1.08,1.37) circle (2.0858571379651094cm);

\draw [] (-8,2)-- (10.68,-2.53);
\draw [red] (0.0908972387793824,-0.3846507539328303)-- (5.82,-4.95);
\draw [blue] (-4.26,2.17)-- (12,-7);

  \foreach \Point in {(0.0908972387793824,-0.3846507539328303),(5.82,-4.95),(-0.56,0.09),(0.83,-0.7),(-7.07,1.77),(4.54,-2.47)}{
   \node at \Point {\textbullet};
}
\draw[] (1.08,1.37) node[] {$\D_1$};
\draw[] (0,-2.25) node[] {$\D_2$};
\draw[] (6.14,-6.33) node[] {$\D_3$};

\draw[above right ,blue] (4.54,-2.47) node[] {$p'_2$};
\draw[above] (-7.07,1.77) node[] {$q$};
\draw[blue, above right] (-0.56,0.09) node[] {$\chi_1$};
\draw[blue, above] (0.83,-0.7) node[] {$\chi_2$};

\draw[red, below left] (0.0908972387793824,-0.3846507539328303) node[] {$p''_1$};
\draw[red, right] (5.82,-4.95) node[] {$p''_3$};
\end{tikzpicture}
    \caption{Illustration of the proof of Lemma~\ref{lemma:notCentred}: $[p''_1,p''_3]$ intersects with $\D_2$, and therefore $\D''$ contains $p'_2$ or is centred with respect to $\D_2$.}
    \label{fig:notcentred}
\end{figure}

%% file: 3disks/2transversal.tex
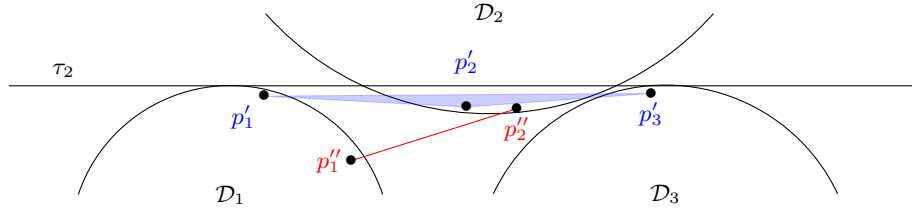
\begin{figure}
    \centering
    \begin{tikzpicture}[scale=0.72]

\draw[ ] (6.14,1.47) arc (25.01:154.99:3.47cm);
\draw[ ] (-2.17,1.46) arc (18.41:161.59:2.93cm);
\draw (-4.31,4.79) arc (221.47:318.53:5.46cm) ;

\draw (-9,3.46)--(7.5,3.46);
\draw[blue,fill, opacity=0.2] (-4.34,3.27)--(2.73,3.32)--(-0.64,3.07)--(-4.34,3.27);
\draw[red] (-2.75,2.08)--(0.28,3.04);

\draw[] (-8,3.46) node[above] {$\tau_2$};

\draw[] (-4.95,1.45) node[] {$\D_1$};
\draw[] (-0.22,4.8) node[] {$\D_2$};
\draw[] (3,1.47) node[] {$\D_3$};

\draw[blue, below left] (-4.34,3.27) node[] {$p'_1$};
\draw[blue, below] (2.73,3.32) node[] {$p'_3$};
\draw[blue] (-0.64,3.9) node[] {$p'_2$};

\draw[red, left] (-2.75,2.08) node[] {$p''_1$};
\draw[red, below] (0.28,3.04) node[] {$p''_2$};

  \foreach \Point in {(-4.34,3.27),(2.73,3.32),(-2.75,2.08),(0.28,3.04),(-0.64,3.07)}{
   \node at \Point {\textbullet};
}

\end{tikzpicture}
    \caption{Illustration of the proof of Lemma~\ref{lemma:twoTransversal}: $\D'$ contains $p''_2$ or is outside-containing $\D_2$.}
    \label{twoTransversal}
\end{figure}

%% file: main.bbl
\begin{thebibliography}{10}

\bibitem{aronov2018pseudo}
Boris Aronov, Anirudh Donakonda, Esther Ezra, and Rom Pinchasi.
\newblock On pseudo-disk hypergraphs.
\newblock {\em arXiv preprint arXiv:1802.08799}, 2018.

\bibitem{bar2001unified}
Amotz Bar-Noy, Reuven Bar-Yehuda, Ari Freund, Joseph Naor, and Baruch Schieber.
\newblock A unified approach to approximating resource allocation and
  scheduling.
\newblock {\em Journal of the ACM (JACM)}, 48(5):1069--1090, 2001.

\bibitem{bonamy2018eptas}
Marthe Bonamy, {\'E}douard Bonnet, Nicolas Bousquet, Pierre Charbit, and
  St{\'e}phan Thomass{\'e}.
\newblock {EPTAS} for max clique on disks and unit balls.
\newblock In {\em 2018 IEEE 59th Annual Symposium on Foundations of Computer
  Science (FOCS)}, pages 568--579. IEEE, 2018.

\bibitem{bonnet2017qptas}
{\'E}douard Bonnet, Panos Giannopoulos, Eun~Jung Kim, Pawe{\l}
  Rz{\k{a}}{\.z}ewski, and Florian Sikora.
\newblock {QPTAS} and subexponential algorithm for maximum clique on disk
  graphs.
\newblock In {\em 34th International Symposiumon Computational Geometry
  (SoCG)}, pages 12:1--12:15, 2018.

\bibitem{bonnet2020maximum}
{\'E}douard Bonnet, Nicolas Grelier, and Tillmann Miltzow.
\newblock Maximum clique in disk-like intersection graphs.
\newblock {\em arXiv preprint arXiv:2003.02583}, 2020.

\bibitem{clark1990unit}
Brent~N Clark, Charles~J Colbourn, and David~S Johnson.
\newblock Unit disk graphs.
\newblock {\em Discrete mathematics}, 86(1-3):165--177, 1990.

\bibitem{dvovrak2020induced}
Zden{\v{e}}k Dvo{\v{r}}{\'a}k and Jakub Pek{\'a}rek.
\newblock Induced odd cycle packing number, independent sets, and chromatic
  number.
\newblock {\em arXiv preprint arXiv:2001.02411}, 2020.

\bibitem{edmonds1972theoretical}
Jack Edmonds and Richard~M Karp.
\newblock Theoretical improvements in algorithmic efficiency for network flow
  problems.
\newblock {\em Journal of the ACM (JACM)}, 19(2):248--264, 1972.

\bibitem{fishkin2003disk}
Aleksei~V Fishkin.
\newblock Disk graphs: A short survey.
\newblock In {\em International Workshop on Approximation and Online
  Algorithms}, pages 260--264. Springer, 2003.

\bibitem{gupta1982efficient}
Udaiprakash~I Gupta, Der-Tsai Lee, and JY-T Leung.
\newblock Efficient algorithms for interval graphs and circular-arc graphs.
\newblock {\em Networks}, 12(4):459--467, 1982.

\bibitem{imai1983finding}
Hiroshi Imai and Takao Asano.
\newblock Finding the connected components and a maximum clique of an
  intersection graph of rectangles in the plane.
\newblock {\em Journal of algorithms}, 4(4):310--323, 1983.

\bibitem{kang2012sphere}
Ross~J Kang and Tobias M{\"u}ller.
\newblock Sphere and dot product representations of graphs.
\newblock {\em Discrete \& Computational Geometry}, 47(3):548--568, 2012.

\bibitem{raghavan2001robust}
Vijay Raghavan and Jeremy Spinrad.
\newblock Robust algorithms for restricted domains.
\newblock In {\em Proceedings of the twelfth annual ACM-SIAM symposium on
  Discrete algorithms}, pages 460--467. Society for Industrial and Applied
  Mathematics, 2001.

\bibitem{van2009optimization}
Erik~Jan van Leeuwen.
\newblock {\em Optimization and approximation on systems of geometric objects}.
\newblock Universiteit van Amsterdam [Host], 2009.

\bibitem{vapnik2015uniform}
Vladimir~N Vapnik and A~Ya Chervonenkis.
\newblock On the uniform convergence of relative frequencies of events to their
  probabilities.
\newblock {\em Theory of Probability and Its Applications}, 16:264--280, 1971.

\bibitem{wenger2017helly}
Rephael Wenger and Andreas Holmsen.
\newblock Helly-type theorems and geometric transversals.
\newblock In {\em Handbook of Discrete and Computational Geometry, Third
  Edition}, pages 91--123. Chapman and Hall/CRC, 2017.

\end{thebibliography}
